\documentclass[twoside,runningheads,10pt]{llncs}
\usepackage[ruled,vlined]{algorithm2e}
\usepackage{amsfonts}
\usepackage{amsmath}
\usepackage{amssymb}
\usepackage{bm}
\usepackage{marvosym}
\usepackage{mathtools}
\usepackage[proofonly,envcountsame]{endproof2}
\usepackage{booktabs}
\usepackage{etoolbox}
\usepackage{filecontents}
\usepackage{float}
\usepackage[para]{footmisc}
\usepackage{graphicx}
\usepackage[bookmarks=true,colorlinks=true,linkcolor=blue,citecolor=blue,urlcolor=blue]{hyperref}
\usepackage{ifsym}
\usepackage{ifthen}
\usepackage{mathrsfs}
\usepackage{paralist}
\usepackage{pgfplots}
\usepackage{prettyref}
\usepackage{subcaption}
\usepackage[only,shortrightarrow,boxdot,boxempty]{stmaryrd}
\usepackage{tabularx}
\usepackage{times}
\usepackage{thmtools}
\usepackage{thm-restate}
\usepackage{tikz}
\usepackage{xcolor}
\usepackage{xspace}
\usepackage{multirow}
\usepackage{mdframed}
\usepackage[outline]{contour}
\usepackage{wrapfig}

\usepackage[setrelation]{math}
\usepackage[prefixflatinterpret,bracketmodalinterpret,fixformat,silentconst,sidenotecalculus,longseqcontext]{logic}
\usepackage[prefixflatinterpret,bracketmodalinterpret,fixformat,silentconst,differentialdL,simplenames]{dL}

\usepackage[irlabel]{fcps}

\usetikzlibrary{decorations}
\usetikzlibrary{decorations.pathmorphing}
\usetikzlibrary{decorations.text}
\usetikzlibrary{decorations.shapes}
\usetikzlibrary{shapes}
\usetikzlibrary{arrows}
\usetikzlibrary{fit}
\usetikzlibrary{calc}
\usetikzlibrary{decorations}
\usetikzlibrary{positioning,shadows}
\usetikzlibrary{automata,trees}
\usetikzlibrary{plotmarks}
\usetikzlibrary{spy,shadows}
\usetikzlibrary{chains,matrix}
\usetikzlibrary{patterns}

\usepgfplotslibrary{fillbetween}

\providecommand{\eg}{e.\,g.\xspace}
\providecommand{\ie}{i.\,e.\xspace}

\DeclareRobustCommand{\todo}[2][\empty]{
  \ifthenelse{\equal{#1}{\empty}}
    {\textsf{\textbf{\color{red}Todo: #2}}}
    {\textsf{\color{red}(#1)\textbf{ Todo: #2}}}
}

\DeclareRobustCommand{\comment}[2][\empty]{
  \ifthenelse{\equal{#1}{\empty}}
    {\textsf{\textbf{\color{green}Comment: #2}}}
    {\textsf{\color{green}(#1)\textbf{ Comment: #2}}}
}

\newcommand{\connectnodes}[2]{\tikz[remember picture,overlay]
\draw[->] (#1) -- (#2);}

\newcommand{\rref}[2][]{\prettyref{#2}}
\newrefformat{algo}{Algorithm\,\ref{#1}}
\newrefformat{app}{Appendix\,\ref{#1}}
\newrefformat{sec}{Section\,\ref{#1}}
\newrefformat{def}{Def.\,\ref{#1}}
\newrefformat{thm}{Theorem\,\ref{#1}}
\newrefformat{cor}{Corollary\,\ref{#1}}
\newrefformat{prop}{Proposition\,\ref{#1}}
\newrefformat{lem}{Lemma\,\ref{#1}}
\newrefformat{ex}{Example\,\ref{#1}}
\newrefformat{tab}{Table\,\ref{#1}}
\newrefformat{fig}{Fig.\,\ref{#1}}
\newrefformat{mod}{Model\,\ref{#1}}
\newrefformat{proof}{Proof\,\ref{#1}}
\newrefformat{step}{Step\,\ref{#1}}
\newrefformat{eq}{Formula\,\eqref{#1}}
\newrefformat{optimization}{Opt.\,\ref{#1}}

\newcommand{\chimp}{\chi_{\text{m}}}
\newcommand{\chimpp}{\chi_{\tilde{\text{m}}}}

\newcommand{\fchig}{F(x,x^+)}

\newcommand{\fchisp}{\fchig}

\newcommand{\hpvars}{V}

\newcommand{\curvebot}{\ensuremath\textit{flight}}
\newcommand{\init}{\ensuremath\textit{A}}
\newcommand{\safe}{\ensuremath\textit{S}}
\newcommand{\fallback}{\ensuremath\beta}
\newcommand{\ctrlPrg}{\ensuremath\textit{ctrl}}
\newcommand{\ctrlOut}{\ensuremath u}
\newcommand{\plantPrg}{\ensuremath\textit{plant}}

\newcommand{\ctrl}[2][u]{\ensuremath{#1}{:}{\in}\textit{ctrl}({#2})}
\newcommand{\senseCtrl}{\ensuremath{\ctrl{\hat{y}}}}
\newcommand{\plant}[2][x]{\ensuremath \D{#1}{=}f({#1},{#2})}
\newcommand{\plantQ}[2][x]{\ensuremath \D{#1}{=}f({#1},{#2})~\&~\ivr}
\newcommand{\plantDur}[3][x]{\ensuremath \plant[#1]{#2}@#3}

\newcommand{\nsimilar}[3]{\ensuremath{#1{\in}\neighborhood{#2}{#3}}}
\newcommand{\prandiv}[3]{\ensuremath{#1{:}{\in}\neighborhood{#2}{#3}}}
\newcommand{\measure}[2]{\ensuremath{\prandiv{\hat{#1}}{#1}{#2}}}
\newcommand{\recallState}[1]{\ensuremath{\humod{#1_0}{#1};\humod{\hat{#1}_0}{\hat{#1}}}}
\newcommand{\estimatorname}{\ensuremath{e}}
\newcommand{\updateEstimator}[4]{\ensuremath{\humod{[#3,#4]}{\estimatorname(\hat{#1}_0,\hat{#1},#1-#1_0,#2,[#3_0,#4_0])}}}
\newcommand{\neighborhood}[2]{\ensuremath{\mathcal{B}_{#2}(#1)}}
\newcommand{\disturbed}[1]{\ensuremath{\tilde{#1}}}
\newcommand{\disturb}[2]{\ensuremath{\prandiv{\disturbed{#1}}{#1}{#2}}}
\newcommand{\marginProp}[4]{\ensuremath{\forall #1{\in}\neighborhood{#2}{#3}#4}}

\newcommand{\violationbound}{\varepsilon}

\newcommand{\Ii}{\dLint[const=I,state=\mu]}
\newcommand{\tIi}{\dLint[const=I,state=\tilde{\mu}]}
\newcommand{\tI}{\dLint[const=I,state=\tilde{\omega}]}
\newcommand{\tIt}{\dLint[const=I,state=\tilde{\nu}]}
\newcommand*{\Iw}[1][]{\dLint[state={#1}]}

\newcommand{\stold}{\iget[state]{\I}}
\newcommand{\stintermediate}{\iget[state]{\Ii}}
\newcommand{\stnew}{\iget[state]{\It}}

\newcommand{\irelmodels}[3]{(\iget[state]{#1},\iget[state]{#2}) \models #3}

\newcommand{\vownship}{\ensuremath{v_o}}
\newcommand{\vintruder}{\ensuremath{v_i}}
\newcommand{\wangular}{\ensuremath{w}}
\newcommand{\wownship}{\ensuremath{\wangular}}
\newcommand{\wintruder}{\ensuremath{\wangular_i}}
\newcommand{\costheta}{\ensuremath{\cos\theta}}
\newcommand{\sintheta}{\ensuremath{\sin\theta}}
\newcommand{\xrel}{\ensuremath{x}}
\newcommand{\yrel}{\ensuremath{y}}
\newcommand{\wDisturbance}{\ensuremath{\Delta}}
\newcommand{\straightinv}{\ensuremath{\mathcal{I}}}
\newcommand{\straightinvlhs}{\ensuremath{\mathcal{S}}}
\newcommand{\straightinvlhsdef}{\ensuremath{\vintruder\sintheta\xrel - (\vintruder\costheta-\vownship)\yrel}}
\newcommand{\straightinvdef}{\ensuremath{\straightinvlhsdef > \vownship+\vintruder}}
\newcommand{\flightinv}[1]{\ensuremath{\mathcal{J}(#1)}}
\newcommand{\flightinvlhs}{\ensuremath{\mathcal{T}}}
\newcommand{\flightinvlhsdef}[1]{\ensuremath{\vintruder #1 \sintheta\xrel - \vintruder #1 \costheta\yrel + \vownship\vintruder\costheta}}
\newcommand{\flightinvdef}[1]{\ensuremath{\flightinvlhsdef{#1} > \vownship\vintruder + \vintruder #1}}
\newcommand{\pilot}{\ensuremath{w_p}}
\newcommand{\curvebotctrl}{\ensuremath{
\pchoice
{\bigl(\humod{\wownship}{0};~\ptest{\straightinv}\bigr)}
{\bigl(\humod{\wownship}{1};~\ptest{\flightinv{\wownship}}\bigr)}}}
\newcommand{\flightUpsilon}{\ensuremath{\xrel^+{=}\xrel \land \yrel^+{=}\yrel \land \theta^+{=}\theta \land \wownship^+{=}\wownship \land \xrel_0^+{=}\xrel_0 \land \yrel_0^+{=}\yrel_0 \land \theta_0^+{=}\theta_0}}

\newcommand{\curvebotplant}{\ensuremath{\{ \D{\xrel} = \vintruder\costheta - \vownship + \wownship\yrel \syssep \D{\yrel} = \vintruder \sintheta - \wownship\xrel \syssep 
      \D{\theta} = -\wownship \}}}

\newfloat{sequentproof}{tbp}{proof}
\floatname{sequentproof}{Proof}

\floatstyle{ruled}
\newfloat{model}{tbp}{model}
\floatname{model}{Model}

\mdfdefinestyle{example}{%
	outerlinewidth=1,%
	leftmargin=0,%
	rightmargin=0,%
	outerlinecolor=black,%
	innertopmargin=2pt,
	innerleftmargin=5pt,%
	innerrightmargin=0,%
	backgroundcolor=white,%
	topline=false,
	leftline=true,
	rightline=false,%
	bottomline=false
}

\definecolor{semblue}{rgb}{0,0,0.7}
\definecolor{vgreen}{rgb}{.1,.5,0}
\definecolor{vred}{rgb}{.7,0,0}
\definecolor{vblue}{rgb}{.1,.15,.62}

\tikzset{
    state/.style={
    		circle,
           draw=black, thick
           },
    plant/.style={
           rectangle,
           rounded corners,
           draw=black, thick,
           minimum height=2em,
           inner sep=2pt,
           text centered,
           },
    every pin/.style={fill=gray!20!white,rectangle,rounded corners=3pt},
    small dot/.style={fill=black,circle,scale=0.3}
}

\tikzstyle{transitionsystem}=[->,>=stealth',shorten >=1pt,auto,semithick,initial text={},
state/.style={circle,very thick,draw=semblue!70,top color=white,bottom color=blue!20,minimum size=2mm,circular drop shadow},
node distance=2cm]

\pgfdeclarelayer{bg}    
\pgfsetlayers{bg,main}  
\tikzstyle{accepting}=[accepting by arrow]

\graphicspath{{./}{fig/}}

\definecolor{lsblue}{HTML}{16303A}
\definecolor{lslightblue}{HTML}{2E6579}
\definecolor{lsverylightblue}{HTML}{4699B9}
\definecolor{lsgreen}{HTML}{5ECEF9}
\definecolor{lslightgreen}{HTML}{54B9DF}
\definecolor{lsrealgreen}{HTML}{2FCC49}
\definecolor{lsred}{HTML}{CC2F4F}
\definecolor{lsochre}{HTML}{B26817}
\definecolor{lsyellow}{HTML}{FFD91C}

\begin{document}
\title{Verified Runtime Model Validation for\\ Partially Observable Hybrid Systems\thanks{
This material is based upon work supported by the National Science Foundation under NSF CAREER Award CNS-1054246 and by AFOSR under grant number FA9550-16-1-0288 and by Defense Advanced Research Projects Agency (DARPA) under grant number FA8750-18-C-0092.
The views and conclusions contained in this document are those of the authors and should not be interpreted as representing the official policies, either expressed or implied, of any sponsoring institution, the U.S. government or any other entity.
}
}
\titlerunning{Verified Runtime Validation for Partially Observable Hybrid Systems}

\author{Stefan Mitsch \and Andr{\'e} Platzer}

\authorrunning{Stefan Mitsch, Andr{\'e} Platzer} 
\institute{Computer Science Department\\Carnegie Mellon University, Pittsburgh PA 15213, USA\\
\email{\{smitsch,aplatzer\}@cs.cmu.edu}
}
\maketitle

\begin{abstract}
Formal verification provides strong safety guarantees but only for \emph{models} of cyber-physical systems.
\emph{Hybrid system models} describe the required interplay of computation and physical dynamics, which is crucial to guarantee what computations lead to safe physical behavior (e.g., cars should not collide).
Control computations that affect physical dynamics must \emph{act in advance} to avoid possibly unsafe future circumstances.
Formal verification then ensures that the controllers correctly identify and provably avoid unsafe future situations \emph{under a certain model of physics}.
But any model of physics necessarily deviates from reality and, moreover, any observation with real sensors and manipulation with real actuators is subject to uncertainty.
This makes runtime validation a crucial step to monitor whether the model assumptions hold for the real system implementation.

The key question is what property needs to be runtime-monitored and what a satisfied runtime monitor entails about the safety of the system: the observations of a runtime monitor only relate back to the safety of the system if they are themselves accompanied by a proof of correctness!
For an unbroken chain of correctness guarantees, we, thus, \emph{synthesize runtime monitors in a provably correct way} from provably safe hybrid system models.
This paper addresses the inevitable challenge of making the synthesized monitoring conditions robust to \emph{partial observability} of sensor uncertainty and \emph{partial controllability} due to actuator disturbance.
We show that the monitoring conditions result in provable safety guarantees with fallback controllers that react to monitor violation at runtime.
\end{abstract}

\section{Introduction}
\label{sec:intro}

Correctness arguments for cyber-physical systems (CPSs) crucially depend on models in order to enable predictions about their future behavior.
Absent any models, neither tests nor verification provide any correctness results except for the limited amount of concrete (test) cases, because they cannot provide predictions for other cases.
Models of physics are crucial in CPSs, but necessarily deviate from reality.
Even cyber components may come with surprises when their detailed implementation is more complex than their verification model.
Linking cyber and physical components with sensors and actuators adds another layer of uncertainty.
These discrepancies are inevitable and call into question how safety analysis results about models can ever transfer to CPS implementations.
Not using models, however, would invalidate all predictions \cite{DBLP:conf/hybrid/PlatzerC07}.

Even though rigorous correct-by-construction approaches promise provably correct implementation of the software portion of the model, their correctness is still predicated on meeting the assumptions of the model about the physical environment and the physical effects of actuators.\footnote{Some implementations deliberately choose to implement the model in a liberal way to interact with unverified components (\eg, use machine learning), or to allow for adaptation at runtime. These implementations require monitoring of the implementation for compliance with its model.}
Whether or not these environment and actuator effects are faithfully represented in the model can only be checked from actual measurements \cite{DBLP:journals/fmsd/MitschP16} during system operation at runtime.

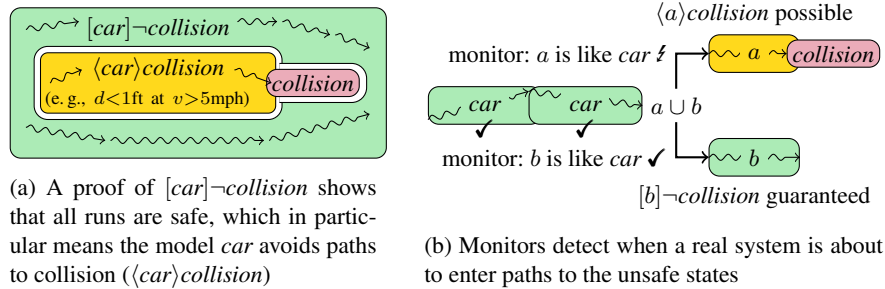
\begin{figure}[tb]
\centering
\pgfdeclarelayer{background}
\pgfdeclarelayer{foreground}
\pgfsetlayers{background,main,foreground}

\begin{subfigure}[b]{.4\linewidth}
\centering
\begin{tikzpicture}
\tikzstyle{snake it} = [decoration={snake,amplitude=.4mm,segment length=2mm,post length=1mm}];
\node[draw,fill=lsyellow,rectangle,rounded corners,double distance=2pt,minimum width=3cm,minimum height=.5cm,text width=3cm] (unsafe-1) {\centering $\didia{\textit{car}}\textit{collision}$\\\begin{scriptsize}(\eg, $d{<}1\text{ft}$ at $v{>}5\text{mph}$)\end{scriptsize}};
\draw[double distance=3pt,rounded corners] ($(unsafe-1.east)+(0,.2cm)$) -- ($(unsafe-1.east)+(1.1cm,.2cm)$) -- ($(unsafe-1.east)+(1.1cm,-.2cm)$) -- ($(unsafe-1.east)+(0,-.2cm)$);
\node[draw,fill=lsred!40,rectangle,rounded corners,inner sep=2pt,right=-5pt of unsafe-1] (unsafe) {$\textit{collision}$};
\draw[->,decorate,snake it] ($(unsafe.west)+(-12pt,5pt)$) -- ($(unsafe.west)+(2pt,0pt)$);
\draw[->,decorate,snake it] ($(unsafe-1.west)+(5pt,0pt)$) -- ($(unsafe-1.west)+(17pt,5pt)$);

\begin{pgfonlayer}{background}
	\node[draw,rectangle,rounded corners,fill=lsrealgreen!40,minimum width=5cm,minimum height=2cm,anchor=east,fit=(unsafe)(unsafe-1)] {};
  \node[above=0pt of unsafe-1] (safetyproof) {$\dbox{\textit{car}}\lnot\textit{collision}$};
  \draw[->,decorate,snake it] ($(safetyproof.north)+(-50pt,-8pt)$) -- ($(safetyproof.north)+(-30pt,-5pt)$);  
  \draw[->,decorate,snake it] ($(safetyproof.north)+(30pt,-5pt)$) -- ($(safetyproof.north)+(45pt,-5pt)$);  
   \draw[->,decorate,snake it] ($(safetyproof.north)+(47pt,-5pt)$) -- ($(safetyproof.north)+(65pt,-10pt)$);  
   \draw[->,decorate,snake it] ($(safetyproof.north)+(67pt,-10pt)$) -- ($(safetyproof.north)+(80pt,-15pt)$);  
   \draw[->,decorate,snake it] ($(safetyproof.north)+(-50pt,-45pt)$) -- ($(safetyproof.north)+(-20pt,-50pt)$);  
   \draw[->,decorate,snake it] ($(safetyproof.north)+(-18pt,-50pt)$) -- ($(safetyproof.north)+(40pt,-50pt)$);  
   \draw[->,decorate,snake it] ($(safetyproof.north)+(42pt,-50pt)$) -- ($(safetyproof.north)+(80pt,-40pt)$);  
\end{pgfonlayer}
\end{tikzpicture}
\caption{A proof of $\dbox{\textit{car}}\lnot\textit{collision}$ shows that all runs are safe, which in particular means the model $\textit{car}$ avoids paths to collision ($\didia{\textit{car}}\textit{collision}$)}\label{fig:overviewsafety}
\end{subfigure}%
\qquad
\begin{subfigure}[b]{.5\linewidth}
\centering
\begin{tikzpicture}
\tikzstyle{snake it} = [decoration={snake,amplitude=.4mm,segment length=2mm,post length=1mm}];
\node[draw,fill=lsrealgreen!40,rectangle,rounded corners,minimum width=1.5cm,minimum height=.5cm] (unsafe-3) {};
\node[draw,fill=lsrealgreen!40,rectangle,rounded corners,minimum width=1.5cm,minimum height=.5cm,right=-5pt of unsafe-3] (unsafe-2) {};
\node[draw,fill=lsyellow,rectangle,rounded corners,inner sep=2pt,minimum width=1.2cm,minimum height=.5cm,above right=5pt and 25pt of unsafe-2] (unsafe-1) {};
\node[draw,fill=lsrealgreen!40,rectangle,rounded corners,inner sep=2pt,minimum width=1.2cm,minimum height=.5cm,below right=5pt and 25pt of unsafe-2] (safe-1) {};
\node[draw,fill=lsred!40,rectangle,rounded corners,inner sep=2pt,right=-5pt of unsafe-1] (unsafe) {$\textit{collision}$};
\draw[->,decorate,snake it] ($(unsafe-3.south west)+(0,2pt)$) -- ($(unsafe-2.north west)+(0,-2pt)$);
\draw[->,decorate,snake it] ($(unsafe-2.north west)+(0,-2pt)$) -- (unsafe-2.east);
\draw[->,decorate,snake it] (unsafe-1.west) -- (unsafe.west);
\draw[->,decorate,snake it] (safe-1.west) -- (safe-1.east);
\node[fill=lsrealgreen!40] at (unsafe-3) {$\textit{car}$};
\node[below=-5pt of unsafe-3] {$\pmb{\checkmark}$};
\node[fill=lsrealgreen!40] at (unsafe-2) {$\textit{car}$};
\node[below=-5pt of unsafe-2] {$\pmb{\checkmark}$};
\node[fill=lsyellow] at (unsafe-1) {$a$};
\node[above=0pt of unsafe-1] (unsafe-1-explain) {$\didia{a}\textit{collision}$ possible};
\node[fill=lsrealgreen!40] at (safe-1) {$b$};
\node[below=0pt of safe-1] (safe-1-explain) {$\dibox{b}\lnot\textit{collision}$ guaranteed};

\node[right=0pt of unsafe-2] (decision) {$a \cup b$};
\draw[->,thick] (decision) |- node[left] {monitor: $a$ is like $\textit{car}~\textbf{\Lightning}$} (unsafe-1);
\draw[->,thick] (decision) |- node[left] {monitor: $b$ is like $\textit{car}~\pmb{\checkmark}$} (safe-1);

\end{tikzpicture}
\caption{Monitors detect when a real system is about to enter paths to the unsafe states}\label{fig:overviewmonitor}
\end{subfigure}

\caption{Logical characterization of safety and monitors}
\label{fig:overview}
\vspace{-\baselineskip}
\end{figure}

The key question, however, is what property needs to be \emph{runtime-monitored on sensor measurements} and what a satisfied monitor implies about the safety of the system in terms of the \emph{possible true values}.
There is a fundamental asymmetry of the power of runtime monitoring in CPS compared to purely discrete software systems.
In pure software it may suffice to monitor the critical property itself and suspend the software upon violation, raising an exception to propagate the violation to surrounding code for mitigation.
\emph{In cyber-physical systems, any such attempt would be fundamentally flawed, because there is no way of reverting time and trying something else!}

A desired property of a self-driving car is to be collision-free (stay outside the red region in \rref{fig:overviewsafety}), which, \eg, might be expressed as always keeping at least 1ft distance to other cars.
Even though this property states a good high-level goal, monitoring whether it will be maintained requires checking membership in a much smaller region that makes it possible to predict that \emph{all} future behavior also respects the property, not just the present state (where a collision may be inevitable on violation).
In our approach,
\begin{itemize}
\item offline safety proofs ensure that a controller avoids all safety violations with respect to an explicit model of physics, sensor uncertainty, and actuator disturbance (see green region in \rref{fig:overviewsafety}, formula $\dbox{\textit{car}}\lnot\textit{collision}$ says that all runs of the car model avoid collision)---this reduces system correctness to the validation question ``did we build the right model?'',
\item formal characterizations of monitor conditions unambiguously and provably correctly describe how to validate a model at runtime and detect discrepancies between the model and the true system from sensor measurements,
\item offline synthesis proofs turn formal monitor characterizations into runtime monitors that act \emph{in advance based on sensor measurements} with provable safety guarantees about the true CPS (see \rref{fig:overviewmonitor}: monitor $\checkmark$ means the system behaves like the model $\textit{car}$ and so we know $\lnot\textit{collision}$ is true in that state from the safety proof; if a monitor determines that a runtime behavior $a$ deviates from the model $\textit{car}$ then we cannot conclude $\dibox{a}\lnot\textit{collision}$, so $\didia{a}\textit{collision}$ is possible and fallback control engaged in a Simplex-style architecture \cite{DBLP:journals/software/Sha01,DBLP:conf/iccps/BakMMC11}), and
\item correctness theorems justify that \emph{all} future model behavior will be safe if the monitor is satisfied at runtime (see \rref{fig:overviewmonitor}: monitor determines that runtime behavior $b$ is like the model $\textit{car}$ whose future behavior is guaranteed collision-free from the safety proof $\dibox{\textit{car}}\lnot\textit{collision}$), and recoverability theorems establish the effectiveness of fallback control on monitor violation.
\end{itemize}

These proofs, formal characterizations, and theorems simultaneously justify model correctness and system safety through a verified link between offline verification and runtime monitoring, and ensure that no assumption is missed that needs to be checked for system correctness.
Such a provably safe monitoring approach crucially requires models of discrete and continuous dynamics and logical foundations for analyzing both necessity $\dbox{\alpha}{\textit{safe}}$ and possibility $\ddiamond{\alpha}{\lnot\textit{safe}}$ in the same framework, which is why we chose differential dynamic logic \cite{Platzer18,DBLP:journals/jar/Platzer17}.
For example, the specification ``keep at least 1ft distance to other cars'' in our approach considers the dynamics and results in a monitor ``always keep stopping distance at least as a specific function of velocity etc.'' that acts ahead of time as opposed to after a collision is already inevitable.
Verified machine code of runtime monitors could be obtained with VeriPhy \cite{DBLP:conf/pldi/BohrerTMMP18}.

Based on ModelPlex \cite{DBLP:journals/fmsd/MitschP16}, this paper addresses the fundamental challenge how such an approach can systematically handle the inevitable complications of sensor uncertainty and actuator disturbance in \emph{partially observable hybrid systems}.
What makes monitoring fundamentally more complex in partial observability cases is that the very state variables that are responsible for the accuracy of the model cannot be measured! The monitor, thus, needs to settle for drawing indirect conclusions about the unobservable state from observations about the evolution of other state variables over time.


\section{Preliminaries: Differential Dynamic Logic by Example}
\label{sec:dl}

This section recalls \emph{differential dynamic logic \dL} \cite{Platzer18,DBLP:journals/jar/Platzer17}, which we use to syntactically characterize the semantic conditions required for correctness of the ModelPlex runtime monitoring approach.
We exploit the proof calculus of \dL \cite{Platzer18,DBLP:journals/jar/Platzer17} to guarantee correctness of the monitors produced for concrete CPS models.
This section also introduces our running example of a simple flight collision avoidance protocol.

\paragraph{Syntax summary.}

\begin{table}[b]
  \newcommand{\foform}{F\xspace}
  \centering
  \caption{Hybrid program (HP) representations of hybrid systems}  
  \begin{tabular}{l@{~~}l@{~}}
    \toprule
       \multicolumn{1}{l@{~~}}{Statement} 
    & \multicolumn{1}{c}{Effect}
    \\
    \midrule 
    $\alpha;\beta$ & sequential composition, first run hybrid program~$\alpha$, then hybrid program~$\beta$  \\
    $\pchoice{\alpha}{\beta}$ 
    	& nondeterministic choice, following either hybrid program~$\alpha$ or $\beta$\\
    $\prepeat{\alpha}$ & nondeterministic repetition, repeats $n\geq 0$ times hybrid program~$\alpha$ \\
    $\pupdate{\umod{x}}{\theta}$ & assign value of term $\theta$ to variable $x$ (discrete jump)\\
    $\pupdate{\umod{x}}{*}$ & assign arbitrary real number to variable $x$\\
    $\ptest{\foform}$ & check that formula $F$ holds in current state, and abort if it does not\\
    $\bigl(\D{x_1}=\theta_1,\dots,$ & evolve $x_i$ along differential equation system $\D{x_i} = \theta_i$\\ 
    \quad{$\D{x_n}=\theta_n ~\&~ \foform\bigr)$} & for any duration within evolution domain~$\foform$\\
    \bottomrule
  \end{tabular}
  \label{tab:hybridprograms}
\end{table}

Differential dynamic logic uses hybrid programs as a notation for hybrid systems (\rref{tab:hybridprograms}).
The set of \dL~formulas is generated by the following grammar (\m{{\sim}\in\{<,\leq,=,\geq,>\}} and $\theta_1,\theta_2$ are arithmetic expressions in~\m{+,-,\cdot,/} over the reals):
\[
\phi ::= \theta_1 \sim \theta_2 \mid \neg \phi \mid \phi \wedge \psi \mid \phi \vee \psi \mid \phi \rightarrow \psi \mid \forall x \phi \mid \exists x \phi \mid \dibox{\alpha}{\phi} \mid \didia{\alpha} \phi
\]

\dL allows us to make statements that we want to be true for all runs of a hybrid program ($\dibox{\alpha}\phi$) or for at least one run ($\didia{\alpha}\phi$).
Both constructs are necessary to derive safe monitors.
We need proofs of $\dibox{\alpha}\phi$ so that we can be sure all behavior of a model are safe. We need proofs of $\didia{\alpha}\phi$ to identify when a system execution can fit to the verified model.
We use \dL's verification technique to prove such correctness properties of hybrid programs \cite{Platzer18,DBLP:journals/jar/Platzer17} as implemented in the \KeYmaeraX prover~\cite{DBLP:conf/cade/FultonMQVP15} .

\paragraph{Example: Horizontal flight collision avoidance.}
We model a simple horizontal collision avoidance protocol for two constant-speed airplanes \cite{Tomlin1998,DBLP:journals/jais/GhorbalJZPGC14}: our controlled ownship takes angular velocity $\wownship$ as pilot commands and can fly straight ($\pumod{\wownship}{0}$) or enter a circular wait pattern ($\pumod{\wownship}{1}$) to avoid collision with a straight-path intruder airplane.

\begin{wrapfigure}[8]{r}[0pt]{3.2cm}
\vspace{-2\baselineskip}
\begin{tikzpicture}
\node (ownship) at (0,0) {\includegraphics[width=1cm]{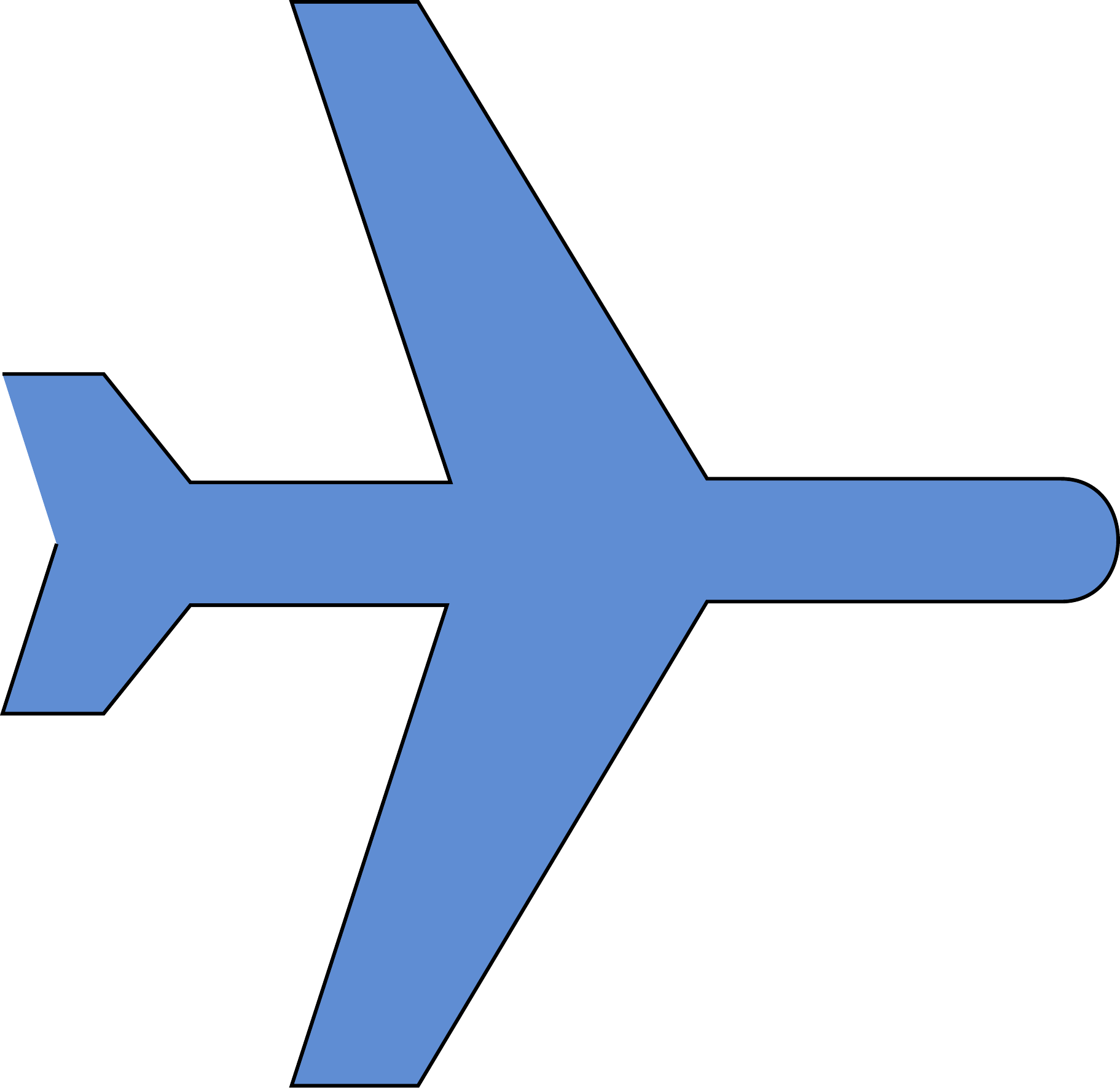}};
\coordinate (intruder) at (1.5,-1.2);
\draw[fill=lsred!20] (intruder.center) -- ($(intruder.center)+(0.7,0)$) arc (0:120:0.7) node[midway,anchor=south west] {$\theta$} -- cycle;
\node at (intruder.center) {\includegraphics[width=1cm,angle=120,origin=c]{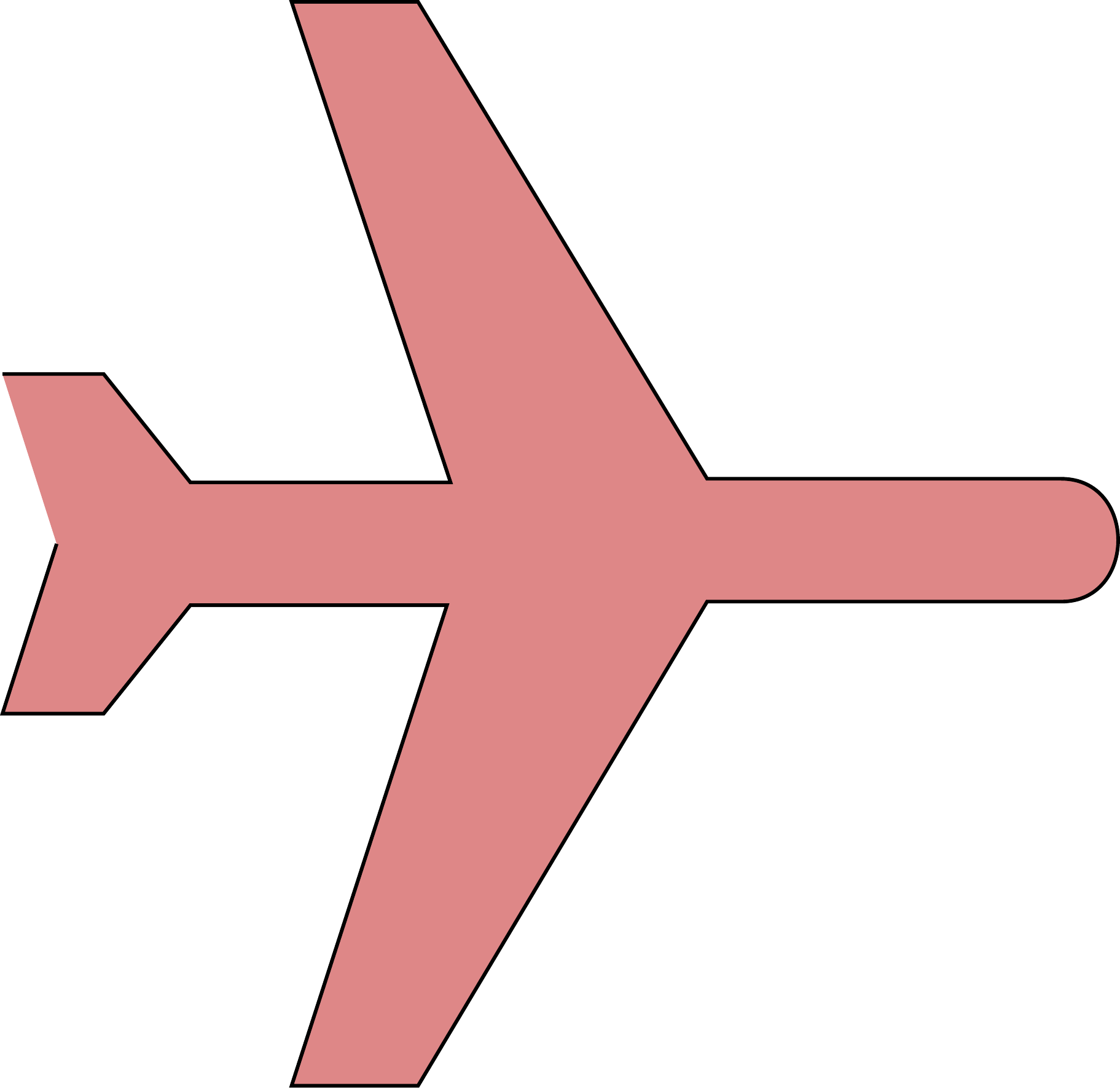}};
\draw[->,very thick] (ownship.center) -- node[anchor=south] {$\vownship$} ($(ownship.center)+(1,0)$);
\draw[->,very thick] (intruder.center) -- node[anchor=east] {$\vintruder$} ($(intruder.center)+(-0.6,1)$);
\draw[->] (ownship.center) -- ($(ownship.center)+(2.5cm,0)$);
\draw[->] (ownship.center) -- ($(ownship.center)+(0,-2cm)$);
\draw[dashed] (intruder.center) -- ($(intruder.center)+(-1.5,0)$);
\draw[dashed] (intruder.center) -- ($(intruder.center)+(0,1.2)$);
\node[anchor=east] at (0,-1.2) {$\yrel$};
\node[anchor=south] at (1.5,0) {$\xrel$};
\end{tikzpicture}
\end{wrapfigure}

The linear and angular velocities of ownship and intruder are independently controlled, but for position and orientation we use a reference frame relative to the ownship centered at $(0,0,0)$ while the intruder is at $(x,y,\theta)$.
The ownship moves with constant linear velocity $\vownship$ and pilot-controlled angular velocity $\wownship$ along a straight line or circle, the intruder with constant linear velocity $\vintruder$ on a straight path.
\begin{align}
\curvebot & \equiv \prepeat{\left(\ctrlPrg;\plantPrg\right)}\label{eq:curvebot:loop}\\
\ctrlPrg &\equiv \curvebotctrl\label{eq:curvebot:ctrl}\\
\straightinv &\equiv \straightinvdef\\
\flightinv{\wangular} &\equiv \flightinvdef{\wangular}\\
\plantPrg &\equiv \curvebotplant\label{eq:curvebot:plant}
\end{align}
The HP $\curvebot$ describes the pilot and  collision avoidance controller $\ctrlPrg$ of the ownship and flight dynamics $\plantPrg$ of both airplanes.
Controller and flight dynamics are repeated nondeterministically often, indicated by $\prepeat{}$ in \eqref{eq:curvebot:loop}.
The pilot has two control choices:
The pilot may choose a straight path $\humod{\wownship}{0}$ if $\ptest{\straightinv}$ indicates that it is safe, or a circular evasion maneuver $\humod{\wownship}{1}$ if allowed by $\ptest{\flightinv{\wownship}}$ in \eqref{eq:curvebot:ctrl}. 
The flight dynamics \eqref{eq:curvebot:plant} keep the (moving) ownship at the origin by combining both ownship and intruder motion in the relative position $(\xrel,\yrel)$; the differential equation $\D{\theta}=-\wownship$ rotates the reference frame.

Formula \eqref{eq:curvebot:safe} specifies safety of the flight protocol $\curvebot$: all runs that start in states satisfying the assumptions $\init$ must stay in states satisfying the safety condition $\safe$.
\begin{equation}\label{eq:curvebot:safe}
\underbrace{\vownship=1 \land \vintruder=1 \land \xrel^2+\yrel^2>0}_{\mathclap{\text{assumptions}~\init}} \limply \dbox{\curvebot}\underbrace{\xrel^2+\yrel^2>0}_{\mathclap{\text{safety}~\safe}}
\end{equation}

\paragraph{Semantics.}
The semantics of \dL \cite{DBLP:journals/jar/Platzer17,DBLP:conf/lics/Platzer12a} is a Kripke semantics in which the states of the Kripke model are the states of the hybrid system, which are maps \m{{\iget[state]{\I}}:{\allvars\to\reals}} assigning a real value $\iget[state]{\I}(x)$ to each variable $x\in\allvars$.
We write $\iget[state]{\I}_x^r$ for the state $\tilde{\iget[state]{\I}}$ that equals $\iget[state]{\I}$ except that $\tilde{\iget[state]{\I}}(x)=r$.
We write~\m{\imodels{\I}{\phi}} if formula~$\phi$ is true in~$\iname[state]{\I}$~$\iget[state]{\I}$.
The semantics of a hybrid program \(\alpha\) is a relation \(\ivaluation{}{\alpha}\) between initial and final states.
For example \(\imodels{\I}{\dbox{\alpha}\phi}\) iff \(\imodels{\It}{\phi}\) for all 
\(\iaccessible[\alpha]{\I}{\It}\), so all runs of $\alpha$ from $\iget[state]{\I}$ are safe.

We use the following notation to refer to variables of hybrid programs and \dL formulas \cite{DBLP:journals/jar/Platzer17}: $\freevars{\alpha}$ and $\freevars{\phi}$ are the free variables, $\boundvars{\alpha}$ and $\boundvars{\phi}$ are the bound variables of a program or formula, respectively, and the complement is $\scomplement{\boundvars{\alpha}}$.
We use $\vars{\alpha} = \freevars{\alpha} \cup \boundvars{\alpha}$ to denote the set of all variables of $\alpha$.

\vspace{-.5\baselineskip}
\paragraph{Notation.}

To concisely handle intervals that we will need for tolerances in partially observable hybrid systems, we use shortcut notation.
We use $\nsimilar{x}{y}{{[l,u]}}$ to say that $x$ is in the interval $[l,u]$ around $y$ (so $y+l \leq x \leq y+u$) and $\nsimilar{x}{y}{\Delta}$ to say $\nsimilar{x}{y}{{[-\Delta,\Delta]}}$.
We use $\prandiv{x}{y}{\Delta}$ to refer to a program that nondeterministically picks any value from the interval $[y-\Delta,y+\Delta]$ for $x$, which is the hybrid program $\prandom{x};~\ptest{y-\Delta \leq x \land x \leq y+\Delta}$. The notation $x{:}{\in}\alpha$ is synonymous with $\alpha$ but emphasizes that the vector of variables $x$ are the only $\boundvars{\alpha}$.
In monitors, we use $x$ to denote the vector of present state variables and $x^+$ for the vector of variables in the next state.
We use $\plantDur{u}{\varepsilon}$ to refer to an ODE that runs for time $\varepsilon$, which is the hybrid program $\humod{t}{0};~\{\plant{u},\D{t}=1 ~\&~t\leq\varepsilon\};~\ptest{t{=}\varepsilon}$.
We abbreviate ODEs to the \emph{interpolated plant effect} $z{-}z_0$ where $z$ means $\stnew(x)$ and $z_0$ means $\stold(x)$ for 
$\iaccessible[\plantDur{\ctrlOut}{\varepsilon}]{\I}{\It}$.

\section{Framework: Monitor Synthesis for Verified Runtime Validation}
\label{sec:approach}

CPS are almost impossible to get right without rigorous safety analysis, for instance by formal verification.
Performed offline, these approaches result in a verified model of a CPS, \ie formula \eqref{eq:generic-safetymodel} is proved valid, for example using the differential dynamic logic proof calculus \cite{DBLP:journals/jar/Platzer17} implemented in  \KeYmaeraX~\cite{DBLP:conf/cade/FultonMQVP15}:
\begin{equation}\label{eq:generic-safetymodel}
\init \limply \dibox{\prepeat{\alpha}}\safe
\end{equation}

The model $\prepeat{\alpha}$ is a hybrid program as in \rref{sec:dl}, and so it describes the discrete control actions of the controllers in the system as well as the continuous physics of the plant and the system's environment.
A formula $\inv$ is an inductive invariant of program $\alpha$ for \eqref{eq:generic-safetymodel} if $A\limply\inv$ and $\inv \limply \dbox{\alpha}\inv$ and $\inv\limply S$ are valid.

Whether or not the control choices, actuator and environment effects are faithfully represented in the model can only be checked from measurements \cite{DBLP:journals/fmsd/MitschP16}. 
Intuitively, such checks compare the values of variables $x$ in state $\nu_{i-1}$ to the values in state $\nu_i$ (stored in $x^+$) taken at successive sample times to check compatibility of the unknown system behavior $\gamma_{i-1}$ with model $\prepeat{\alpha}$ (\rref{fig:approach}).

\begin{wrapfigure}[9]{r}[0pt]{.45\textwidth}
\vspace{-5ex}
\begin{tikzpicture}[transitionsystem,
ctrlblock/.style={
	draw,dashed,rounded corners,fill=gray!40}
]
  \node[initial,state] (initial)                      {};
  \node[state,minimum size=1cm]         	(n0) [right of=initial,xshift=-0.5cm] {$\nu_{i-1}$};
  \node[state,minimum size=1cm]         	(n1) [right of=n0,xshift=0.5cm] {$\nu_{i}$};
  
  \draw (initial) -- node[above] {$\gamma_{i-2}$} node[below] {$\subseteq \prepeat{\alpha}$} (n0);
  \draw (n0) -- node [above] {$\gamma_{i-1}$} node [below] {$\stackrel{?}{\subseteq}\prepeat{\alpha}$}(n1);  
  
  \draw[-,decorate,decoration={brace,amplitude=3pt,mirror}] 
    ($(n0.south)+(0.1,-0.2)$) -- node [below=0.1,text width=5.2cm,align=center,xshift=-0.8cm] {\textbf{Model monitor:} measurements and control decisions agree with model?}
    ($(n1.south)+(-0.1,-0.2)$); 	
\end{tikzpicture}
\caption{System transition $\gamma_i$ checked for validation against model transition $\prepeat{\alpha}$.}
\label{fig:approach}
\end{wrapfigure}

For example, measurements $x{=}2$ and $x^+{=}3$ are compatible with the repeated differential equation $\prepeat{(\D{x}{=}x^2+x)}$, because starting at $x$ can produce $x^+$, as witnessed by a proof of the \dL formula $x{=}2 \land x^+{=}3 \limply \didia{\prepeat{(\D{x}{=}x^2+x)}}(x^+{=}x)$. 
No $x^+ < x$ is compatible with the program, since the program can reach only states where $x^+\geq x$.

ModelPlex~\cite{DBLP:journals/fmsd/MitschP16} provides the basis for obtaining such runtime checks from hybrid system models both \emph{automatically} and in a \emph{provably correct manner} (\rref{app:arithmeticalform} shows a basic example how a ModelPlex proof synthesizes runtime checks). 
ModelPlex simplifies analysis in a verified manner to the loop body $\alpha$ of a hybrid program $\prepeat{\alpha}$, starting from monitor conditions in \dL of the form $\didia{\alpha}\Upsilon^+$, where $\Upsilon^+$ is shortcut notation for $x^+=x$ for all $x \in \boundvars{\alpha}$ to collect the effect of executing $\alpha$ (with a focus on preserving proof-relevant dynamics, not necessarily the exact trajectories as in \cite{DBLP:conf/cav/SankaranarayananT11}).
For a previous state $\stold$ and new state $\stnew$ does the satisfaction relation $\irelmodels{\I}{\It}{\phi}$ hold iff monitor condition $\phi$ holds with its $x$ values coming from $\stold$ and $x^+$ values from $\stnew$ \cite{DBLP:journals/fmsd/MitschP16}:

\begin{definition}[Transition satisfaction relation]
\label{def:monitorevaluation}
The satisfaction relation $\irelmodels{\I}{\It}{\phi}$ of \dL formula $\phi$ for a pair of states $(\stold,\stnew)$ evaluates $\phi$ in the state resulting from state $\stold$ by interpreting variable $x^+$ as $\stnew(x)$ for all $x \in \allvars$, \ie,
$\irelmodels{\I}{\It}{\phi}$ iff $\imodels{\Iw[\stold_{x^+}^{\stnew(x)}]}{\phi}$.
Otherwise we write $(\stold,\stnew) \not\models \phi$.
\end{definition}

The central correctness result about ModelPlex \cite[Theorems 1 and 2]{DBLP:journals/fmsd/MitschP16} guarantees that its resulting monitoring formula preserves safety, \ie, if the monitoring formula is true with current sensor measurements and control choices, then the system is safe.
Here, we complement \cite[Theorems 1 and 2]{DBLP:journals/fmsd/MitschP16} with recoverability guarantees provided by fallback control per \rref{def:fallbackcontrol} when the monitors are violated.
Recoverability guarantees are phrased in terms of the inductive invariant $\inv$ used in the safety proof of $A \limply \dbox{\prepeat{\alpha}}S$, since $\inv$ implies present safety ($\inv \limply S$) and is maintained in the future ($\inv \limply \dibox{\alpha}\inv$).

\newcommand{\ctrlmonitordl}{\ensuremath{\didia{\ctrl{x}}\Upsilon^+}}
\begin{definition}[Fallback control]\label{def:fallbackcontrol}
A fallback control is any hybrid program $\fallback$ such that for all states $\stold$,$\stnew$ with $\iaccessible[\fallback]{\I}{\It}$ we have $\irelmodels{\I}{\It}{\ctrlmonitordl}$.
\end{definition}

If the actions that led to monitor violation do not have permanent physical effect (\eg, they can be discarded or undone instantaneously, such as wrong control decisions before they are handed to actuators), then monitor violation is recoverable by fallback:

\begin{theorem}[Control violation recoverability]\label{thm:controlviolationrecoverability}
Let $\inv$ be an inductive invariant for program $\prepeat{(\ctrl{x};\plant{\ctrlOut})}$.
Let $\imodels{\I}{\inv}$ with state $\stintermediate$ have a monitor violation\linebreak $(\stold,\stintermediate) \not\models \ctrlmonitordl$.
Let $\tilde{\stintermediate}$ be a state recovered from monitor violation by fallback control $\fallback$, \ie, $\iaccessible[\fallback]{\I}{\Iw[\tilde{\stintermediate}]}$.
Then $\imodels{\It}{\inv}$ for all states $\stnew$ with $\iaccessible[\plant{\ctrlOut}]{\Iw[\tilde{\stintermediate}]}{\It}$.
\end{theorem}

\begin{proof}
See \rref{app:proofs}.
\end{proof}
\begin{proofatend}
From $\iaccessible[\fallback]{\I}{\Iw[\tilde{\stintermediate}]}$ and so $\irelmodels{\I}{\Iw[\tilde{\stintermediate}]}{\ctrlmonitordl}$ by assumption, we get $\iaccessible[\ctrl{x}]{\I}{\Iw[\tilde{\stintermediate}]}$
by \cite[Thm. 2]{DBLP:journals/fmsd/MitschP16}. 
Now $\iaccessible[\ctrl{x}]{\I}{\Iw[\tilde{\stintermediate}]}$ and in turn $\iaccessible[\ctrl{x};\plant{\ctrlOut}]{\I}{\It}$ from the semantics of sequential composition with assumption $\iaccessible[\plant{\ctrlOut}]{\Iw[\tilde{\stintermediate}]}{\It}$.
Hence we conclude $\imodels{\It}{\inv}$ by assumption $\inv \limply \dbox{\ctrl{x};\plant{\ctrlOut}}\inv$ with assumption $\imodels{\I}{\inv}$.
\end{proofatend}

By \rref{thm:controlviolationrecoverability}, control violations are recoverable by replacing unsafe actions with fallback before they take effect.
Model violations, which are observed on the physical effects, are detected at the earliest point in time, so would be recoverable in the model if the fallback were to interfere at the start of the violation, see \rref{thm:modelviolationrecovarability}.

\newcommand{\modelmonitordl}{\ensuremath{\didia{\alpha}\Upsilon^+}}
\begin{theorem}[Model violation recoverability]\label{thm:modelviolationrecovarability}
Let $\inv$ be an inductive invariant for program $\alpha \equiv \ctrl{x};\plant{\ctrlOut}$.
Let $\imodels{\I}{\inv}$ satisfy monitor $\irelmodels{\I}{\Iw[\stintermediate]}{\modelmonitordl}$, but $(\stintermediate,\stnew) \not\models \modelmonitordl$ violate it.
Let $\tilde{\stintermediate}$ be a state recovered from monitor violation by fallback control $\fallback$ from $\stintermediate$, \ie, $\iaccessible[\fallback]{\Iw[\stintermediate]}{\Iw[\tilde{\stintermediate}]}$.
Then $\imodels{\Iw[\tilde{\stnew}]}{\inv}$ for all states $\tilde{\stnew}$ with $\iaccessible[\plant{\ctrlOut}]{\Iw[\tilde{\stintermediate}]}{\Iw[\tilde{\stnew}]}$.
\end{theorem}

\begin{proof}
See \rref{app:proofs}.
\end{proof}
\begin{proofatend}
From \cite[Thm. 1]{DBLP:journals/fmsd/MitschP16} with $\irelmodels{\I}{\Iw[\stintermediate]}{\modelmonitordl}$ we get $\imodels{\Iw[\stintermediate]}{\inv}$ as $\imodels{\I}{\inv}$.
Now $\iaccessible[\fallback]{\Iw[\stintermediate]}{\Iw[\tilde{\stintermediate}]}$ by assumption and so $\irelmodels{\Iw[\stintermediate]}{\Iw[\tilde{\stintermediate}]}{\ctrlmonitordl}$ by \rref{def:fallbackcontrol}. 
Hence we get $\iaccessible[\ctrl{x}]{\Iw[\stintermediate]}{\Iw[\tilde{\stintermediate}]}$ by \cite[Thm. 2]{DBLP:journals/fmsd/MitschP16} and in turn\linebreak $\iaccessible[\ctrl{x};\plant{\ctrlOut}]{\Iw[\stintermediate]}{\Iw[\tilde{\stnew}]}$ by the semantics of sequential composition with assumption $\iaccessible[\plant{\ctrlOut}]{\Iw[\tilde{\stintermediate}]}{\Iw[\tilde{\stnew}]}$.
Now $\iaccessible[\ctrl{x};\plant{\ctrlOut}]{\Iw[\stintermediate]}{\Iw[\tilde{\stnew}]}$ and so we conclude $\imodels{\Iw[\tilde{\stnew}]}{\inv}$ from assumption $\inv \limply \dbox{\ctrl{x};\plant{\ctrlOut}}\inv$ with $\imodels{\Iw[\stintermediate]}{\inv}$.
\end{proofatend}

\rref{thm:controlviolationrecoverability} and \ref{thm:modelviolationrecovarability} show how fallback control guarantees safety when monitor conditions are violated at runtime.
The monitor conditions in their \dL representation $\didia{\alpha}\Upsilon^+$ are transformed for execution into a real arithmetic representation by an offline proof with the ModelPlex process~\cite{DBLP:journals/fmsd/MitschP16}, recalled in \rref{app:implementation}, and further into machine code using VeriPhy \cite{DBLP:conf/pldi/BohrerTMMP18}.
Our focus in the following sections is to advance the \dL monitor conditions and their correctness proofs to account for crucial actuator disturbance and sensor uncertainty in \emph{partially observable} systems, reusing the same synthesis process.

\section{Monitors for Partial Controllability with Actuator Disturbance}

\newcommand{\disturbancemonitor}[1]{\ensuremath{\ddiamond{#1}\exists \disturbed{\ctrlOut}^+\,\Upsilon^+}}

Actuator disturbance results in discrepancies between the chosen control decisions $\ctrlOut$ and their physical effect $\disturbed{\ctrlOut}$ (\eg, wheel slip).
\rref{def:disturbancenf} introduces a typical pattern to model piecewise constant actuator disturbance, which chooses a nondeterministic value $\disturbed{\ctrlOut}$ in the $\Delta$-neighborhood around the control choice $\ctrlOut$ as input to the plant. 

\begin{definition}[Disturbance normal form]\label{def:disturbancenf}
A hybrid program $\alpha$ in \emph{$\Delta$-disturbance normal form} has the shape $\ctrl{x};~ \disturb{\ctrlOut}{\Delta};~\plant{\disturbed{\ctrlOut}}$.
\end{definition}

Actuator disturbance is partially observable by monitoring the difference between the intended effect $\plant{\ctrlOut}$ and the actual effect $\plant{\disturbed{\ctrlOut}}$ from observable quantities.
By partial observability, safety and invariant properties do not mention the perturbed actuator effect.
For example, a car controller can estimate deviation from its acceleration choice only after the fact by observing the speed difference.
We therefore adapt $\didia{\alpha}\Upsilon^+$ to conjecture existence of an actuator effect $\disturbed{\ctrlOut}^+$ to explain all other effects of program $\alpha$ collected in $\Upsilon^+$. 
Monitor condition \eqref{eq:actuatormonitor} preserves safety by \rref{thm:mm-disturbance}.
\begin{equation}\label{eq:actuatormonitor}
\disturbancemonitor{\alpha}
\end{equation} 

\begin{theorem}[Monitor with actuator disturbance preserves invariants]
\label{thm:mm-disturbance}
Let $\alpha$ be a hybrid program in $\Delta$-disturbance normal form with an inductive invariant $\inv$ where $\disturbed{\ctrlOut} \not\in \freevars{\inv}$.
Assume the system transitions from state $\stold$ to $\stnew$, which agree on $\scomplement{\boundvars{\alpha}}$, and assume $\imodels{\I}{\inv}$.
If the monitor condition \eqref{eq:actuatormonitor} is satisfied, \ie, $\irelmodels{\I}{\It}{\disturbancemonitor{\alpha}}$,
then the invariant is preserved, \ie, $\imodels{\It}{\inv}$.
\end{theorem}

\begin{proof}
By logical state relation and coincidence, see \rref{app:proofs}.
\end{proof}
\begin{proofatend}
The monitor condition $\irelmodels{\I}{\It}{\disturbancemonitor{\alpha}}$ is satisfied by assumption, so we get $\irelmodels{\I}{\It}{\exists \disturbed{x}^+ \didia{\alpha}\Upsilon^+}$ by Barcan \cite{DBLP:journals/jar/Platzer17} with $\disturbed{x}^+ \not\in \vars{\alpha}$.
Hence, there exists a state $\stnew_{\disturbed{x}}^r$ with $r \in \mathbb{R}$ for $\disturbed{x}$ such that $\irelmodels{\I}{\Iw[\stnew_{\disturbed{x}}^r]}{\didia{\alpha}\Upsilon^+}$ and so $\iaccessible[\alpha]{\I}{\Iw[\stnew_{\disturbed{x}}^r]}$ by \rref{lem:logicalrelation}.
Now $\imodels{\I}{\inv}$ and $\models \inv \limply \dbox{\alpha}\inv$ by assumption and therefore $\imodels{\I}{\dbox{\alpha}\inv}$ and in turn $\imodels{\Iw[\stnew_{\disturbed{x}}^r]}{\inv}$ by $\iaccessible[\alpha]{\I}{\Iw[\stnew_{\disturbed{x}}^r]}$.
Since $\stnew = \stnew_{\disturbed{x}}^r$ on $\scomplement{\{\disturbed{x}\}}$ we conclude $\imodels{\It}{\inv}$ by \rref{lem:coincidence} with $\disturbed{x} \not\in \freevars{\inv}$.
\end{proofatend}

Monitoring with actuator disturbance preserves invariant conditions of a system by \rref{thm:mm-disturbance}, which, with Theorems \ref{thm:controlviolationrecoverability}+\ref{thm:modelviolationrecovarability}, \emph{guarantees safety at the present moment as well as safety of all future behavior that fits to the model}.

\begin{mdframed}[style=example]
\begin{example}
In the flight example, the $\plantPrg$ follows decisions precisely without disturbance.
Here, we model a pilot decision $\pilot$ which is subject to disturbance $\wDisturbance\geq 0$ before taking effect to analyze safety for imperfectly actuated evasion maneuvers:
\[
\alpha \equiv \bigl(\pchoice{(\humod{\pilot}{0};\ptest{\straightinv})}{(\humod{\pilot}{1};\ptest{\flightinv{\pilot}})}\bigr); \prandom{\wownship};\ptest{(0{\leq}\wownship{\leq}\pilot \wDisturbance)};~\plantPrg
\]
The true acceleration $\wownship$ is unobservable, so the monitor condition $\didia{\alpha}\exists \wownship^+\,(\flightUpsilon \land \pilot^+{=}\pilot)$ existentially quantifies away the unobservable $\wownship^+$.
\end{example}
\end{mdframed}

\section{Monitors for Partial Observability from Sensor Uncertainty}
\label{sec:sensoruncertainty}

Monitor correctness \cite{DBLP:journals/fmsd/MitschP16} requires that all bound variables of a program $\alpha$ be monitored, \ie, $\boundvars{\alpha} \subseteq \freevars{\Upsilon^+}$.
This is the appropriate behavior except, of course, for variables that are \emph{unobservable} in the CPS implementation.
When controllers use sensors to obtain information, only the measurement is known, but not the true value that the sensor is measuring, which would defeat the purpose.
There may still be indirect implications about unobservable quantities when they relate to observable ones, which necessitates monitors that indirectly check the properties of the true quantity from measurements.

\begin{figure}[tb]
\centering
\begin{subfigure}[b]{.48\linewidth}
\centering
\begin{footnotesize}
\begin{tikzpicture}
\begin{axis}[
  compat=newest,x dir=normal,
  width=\columnwidth,height=4cm,
  xlabel={t},ylabel={x},
  xticklabels=\empty,yticklabels=\empty,
  axis lines=middle,
  axis line style={->},
  x label style={at={(current axis.right of origin)},anchor=west},
  y label style={at={(current axis.above origin)},anchor=east},
  enlargelimits=true
]
\addplot+[color=lslightblue, ultra thick, solid, smooth, mark=*,nodes near coords,nodes near coords align={north west},every node near coord/.append style={font=\scriptsize},point meta=explicit symbolic
]
 plot [error bars/.cd, y dir = both, y explicit, error bar style={line width=2pt,solid}]
 table[meta=xlabel, x =t, y =x, y error expr={0.04}]{modelmonitorpairwise.dat};
\end{axis}
\end{tikzpicture}
\end{footnotesize}
\caption{Safety proof: measurements $\hat{x}$ are taken near true $x$ (bars indicate uncertainty) but system model behavior follows true $x$.}
\label{fig:sensoruncertaintycontroller}
\end{subfigure}%
\hfill%
\begin{subfigure}[b]{.48\linewidth}
\centering
\begin{footnotesize}
\begin{tikzpicture}
\begin{axis}[
  compat=newest,x dir=normal,
  width=\columnwidth,height=4cm,
  xlabel={t},ylabel={x},
  xticklabels=\empty,yticklabels=\empty,
  axis lines=middle,
  axis line style={->},
  x label style={at={(current axis.right of origin)},anchor=west},
  y label style={at={(current axis.above origin)},anchor=east},
  enlargelimits=true
]
\addplot+[color=lslightblue, only marks, mark=square,nodes near coords,nodes near coords align={west},every node near coord/.append style={font=\scriptsize},point meta=explicit symbolic
]
 plot [error bars/.cd, y dir = both, y explicit, error bar style={line width=2pt,solid}]
 table[meta=mxlabel, x =t, y =x, y error expr={0.04}]{modelmonitorpairwise.dat};
 \addplot+[name path=trueupper,mark=none,dotted,smooth,color=lsblue,forget plot] table[x =t, y expr = {\thisrow{x} + 0.04}]{modelmonitorpairwise.dat}; 
 \addplot+[name path=truelower,mark=none,dotted,smooth,color=lsblue,forget plot] table[x =t, y expr = {\thisrow{x} - 0.04}]{modelmonitorpairwise.dat}; 
\addplot[lslightblue!50,opacity=0.5] fill between[of=truelower and trueupper];
\begin{scope}[shorten >=2pt]
\draw[->,thick,dashed] (0,0.53) -- (1,0.58);
\draw[->,thick,dashed] (0,0.53) -- (1,0.56);
\draw[->,thick,dashed] (0,0.53) -- (1,0.54);
\draw[->,thick,dashed] (0,0.48) -- (1,0.53);
\draw[->,thick,dashed] (0,0.48) -- (1,0.51);
\draw[->,thick,dashed] (2,0.62) -- (3,0.65);
\draw[->,thick,dashed] (2,0.67) -- (3,0.59);
\draw[->,thick,dashed] (2,0.67) -- (3,0.66);
\end{scope}
\end{axis}
\end{tikzpicture}
\end{footnotesize}
\caption{Monitor: estimates true $x$ from sample measurements $\hat{x}$, considering $\hat{x}$ plausible if model behavior fits to \emph{some} possible true $x$.}
\label{fig:sensoruncertaintymonitor}
\end{subfigure}
\caption{(a) Controllers observe true behavior through sensors. (b) Monitors have to check existence of behavior that fits to the model and explains the measurements.}
\label{fig:sensoruncertaintyproofvsmonitor}
\end{figure}
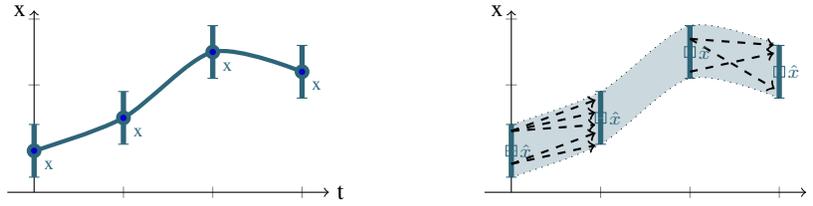

Unobservability results in a crucial difference between monitoring and control (and its safety proofs). 
Controllers \emph{estimate true behavior} by taking measurements and observing the effect of their decisions in the next measurements that are again taken from the subsequent true values, see \rref{fig:sensoruncertaintycontroller}.
Monitors \emph{have to decide} whether the measurements explain some possible true behavior that fits to the expected model behavior, see \rref{fig:sensoruncertaintymonitor}, which results in a number of challenges that we address in this section: 
\begin{inparaenum}[(i)]
\item check existence of behavior as explanation from a \emph{set of possible true values} into a set of possible next true values,
\item \emph{link explanations} to a full path through sets of possible true values around a history of observations, and
\item \emph{guarantee safety} in both cases.
\end{inparaenum}

We use $y$ to refer to an \emph{unobservable state variable} and $\hat{y}$ to denote a \emph{measurement} of $y$ with some uncertainty $\Delta$, so $\nsimilar{\hat{y}}{y}{\Delta}$.
We assume \emph{non-faulty sensors} that function according to specification ($\imodels{\I}{y{\in}\neighborhood{\hat{y}}{\Delta}}$ in any state $\stold$), \ie, they always report values that deviate from the true values by at most some known uncertainty $\Delta$.\footnote{%
Sensor fusion detects sensor faults and corrects measurement outliers to satisfy this assumption.}
\rref{def:uncertaintysimilarity} captures what it means for states to be similar with respect to measurement uncertainty.

\newcommand{\mstateapprox}[4]{\ensuremath{#1 \hat{\approx}_#3^{#4} #2}}
\newcommand{\sensormonitordef}[1]{\ensuremath{\exists \nsimilar{y}{\hat{y}}{\Delta}\,\didia{#1}\left(\exists y^+\,\Upsilon^+\right)}}

\begin{definition}[Uncertainty similarity]
\label{def:uncertaintysimilarity}
We say that state $\stold$ is \emph{$\Delta$-uncertainty-similar on $y$} to state $\stnew$, denoted \(\mstateapprox{\stold}{\stnew}{y}{\Delta}\), iff $\stold=\stnew$ on $\scomplement{\{y\}}$ and $\nsimilar{\stold(y)}{\stnew(\hat{y})}{\Delta}$.
\end{definition}

\rref{def:uncertaintysimilarity} together with measurement uncertainty $\nsimilar{\hat{y}}{y}{\Delta}$ implies that the possible values $\stold(y)$ are at most $2\Delta$ from the true $\stnew(y)$, so $\nsimilar{\stold(y)}{\stnew(y)}{2\Delta}$.

In the following subsections, we characterize monitors that check whether or not there exist states that are $\Delta$-uncertainty-similar to measured states and connected through a program $\alpha$, \ie, for measured states $\iget[state]{\I}$ and $\iget[state]{\It}$ do there exist uncertainty-similar states $\iget[state]{\tI}$ and $\iget[state]{\tIt}$ such that $\mstateapprox{\iget[state]{\tI}}{\iget[state]{\I}}{y}{\Delta}$, $\mstateapprox{\iget[state]{\tIt}}{\iget[state]{\It}}{y}{\Delta}$ and $\iaccessible[\alpha]{\tI}{\tIt}$.
The theorems will be phrased for a single variable $y$ and measurement $\hat{y}$, but extend to vectors of unobservable variables $\vec{y}$ and their measurements $\vec{\hat{y}}$ in a straightforward way.

\subsection{Model Monitors for Pairwise Consistency of Measurements}

\begin{wrapfigure}[13]{r}[0pt]{.4\columnwidth}
\vspace{-2.7\baselineskip}
\centering
\begin{footnotesize}
\vspace{\baselineskip}
\begin{tikzpicture}
\begin{axis}[
  compat=newest,x dir=normal,
  legend style={at={(0.4,1.1)},anchor=north},
  width=.5\columnwidth,height=4.5cm,
  xlabel={t},ylabel={x},
  xticklabels=\empty,yticklabels=\empty,
  axis lines=middle,
  axis line style={->},
  x label style={at={(current axis.right of origin)},anchor=west},
  y label style={at={(current axis.above origin)},anchor=east},
  enlargelimits=true
]
\addplot+[color=lslightblue, ultra thick, solid, mark=square*,
]
 plot [error bars/.cd, y dir = both, y explicit, error bar style={line width=2pt,solid}]
 table[x =t, y =v, y error =ds]{curvebotmodelmonitorpairwise.dat};
\addplot [color=lsblue, thick, only marks, mark=*
]
 plot [error bars/.cd, y dir = both, y explicit
 ]
 table[x =t, y =snow, y error =ds]{curvebotmodelmonitorpairwise.dat}; 

\addplot+[name path=trueupper,mark=none,dashed,smooth,very thick,color=lsblue,forget plot] table[x =t, y expr = {\thisrow{v} + \thisrow{ds}}]{curvebotmodelmonitorpairwise.dat}; 
\addplot+[name path=truelower,mark=none,dashed,smooth,very thick,color=lsblue,forget plot] table[x =t, y expr = {\thisrow{v} - \thisrow{ds}}]{curvebotmodelmonitorpairwise.dat};
\addplot+[name path=measureupper,mark=none,dotted,smooth,color=lsblue,forget plot] table[x =t, y expr = {\thisrow{snow} + \thisrow{ds}}]{curvebotmodelmonitorpairwise.dat}; 
\addplot+[name path=measurelower,mark=none,dotted,smooth,color=lsblue,forget plot] table[x =t, y expr = {\thisrow{snow} - \thisrow{ds}}]{curvebotmodelmonitorpairwise.dat}; 
\addplot[lsochre!35,opacity=0.5,
] fill between[of=measurelower and measureupper];

\addlegendentry{True $y$};
\addlegendentry{Measured $\hat{y}$};
\end{axis}
\end{tikzpicture}
\vspace{-2\baselineskip}
\end{footnotesize}
\caption{Measurements with sensor uncertainty. Thick blue bars: measurements $\nsimilar{\hat{y}}{y}{\Delta}$; thin black bars: estimated true $\nsimilar{y}{\hat{y}}{\Delta}$.}
\label{fig:sensoruncertaintymodelmonitor}
\end{wrapfigure}
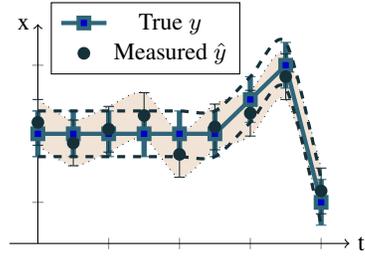%
Monitoring based on measurements requires us to decide whether a true value $y$ fits to a model $\alpha$ by only looking at the measurement $\hat{y}$. 
Intuitively, this can be answered by finding an unobservable prior state $\nsimilar{y}{\hat{y}}{\Delta}$ close to the measurement $\hat{y}$, such that running the model $\alpha$ on this possible $y$ predicts a next unobservable $y^+$ that is within measurement uncertainty $y^+{\in}\neighborhood{\hat{y}^+}{\Delta}$ to the next measurement $\hat{y}^+$.
This intuition is illustrated in \rref{fig:sensoruncertaintymodelmonitor}:
a pair of two consecutive measurements is possible if the set of possible true values of the second measurement overlap with the values predicted by the model $\alpha$ from one of the possible prior $y$ (which is estimated from the previous measurement). 
\rref{def:measurementnf} captures this intuition with a hybrid program that produces control output $\ctrlOut$ from a previous measurement to drive a plant $\plantDur{\ctrlOut}{\varepsilon}$ for time $\varepsilon$ to produce the next measurement.

\begin{definition}[Measurement normal form]\label{def:measurementnf}
A hybrid program $\alpha$ in \emph{$\Delta$-measurement normal form} has the shape $\senseCtrl;
\plantDur{\ctrlOut}{\varepsilon};~\measure{y}{\Delta}$.
\end{definition}

Our goal when formalizing monitoring conditions is to shift proof effort offline in favor of fast runtime computations. 
Therefore we combine offline proofs with online monitoring:
\begin{inparaenum}[(i)]
\item offline we prove that the modeled dynamics ensure that \emph{all possible true values} around the measurements satisfy the invariant (contraction, see \rref{def:contraction}), and
\item online we monitor that \emph{some possible true values} around the successive measurements can be interconnected via the modeled dynamics, see monitor condition \eqref{eq:mm-pairwiseuncertainty}.
\end{inparaenum}

\begin{definition}[Contraction]\label{def:contraction}
The \emph{$[l,u]$-contraction} with margin $[l,u]$, which is\linebreak $\marginProp{y}{\hat{y}}{[l,u]}{\inv}$, ensures $\inv$ for all possible values $y$ in the $[l,u]$-neighborhood of the value $\hat{y}$.
A program $\alpha$ in $\Delta$-measurement normal form is \emph{$[l,u]$-contraction-safe} if\linebreak $(\marginProp{y}{\hat{y}}{[l,u]}{\inv}) \limply \dibox{\alpha}\marginProp{y}{\hat{y}}{[l,u]}{\inv}$ is valid. 
\end{definition}

The $[l,u]$-contraction $\marginProp{y}{\hat{y}}{[l,u]}{\inv}$ requires the controller to consider all possible true values $y$ around the measurement $\hat{y}$ when it computes its control output $\ctrlOut$. 
This makes control decisions robust to measurement uncertainty, because the control output $\ctrlOut$ is required to be chosen such that all possible values around the next measurement again satisfy $\inv$.
Geometrically, $\marginProp{y}{\hat{y}}{[l,u]}{\inv}$ corresponds to $\stold(\hat{y}) \oplus [l,u] \subseteq \iaccess[\inv]{}$ for the Minkowski sum $\oplus$.
If a controller does not account for the difference between the true physical dynamics and the modeled dynamics, it may approach the safety boundary too aggressively and potentially result in unsafe behavior.
\rref{def:contraction} precisely captures controller robustness to uncertainty and perturbation in terms of the physical effect of its control choices, which enables monitor correctness per \rref{thm:mm-pairwisecorrectness} later.
Designing a controller for the robustness criterion of \rref{def:contraction} is an important research field outside the scope of this paper.
Monitor condition \eqref{eq:mm-pairwiseuncertainty} checks two measurements $\hat{y}$ and $\hat{y}^+$ according to a model $\alpha$ in $\Delta$-measurement normal form. 
\begin{equation}\label{eq:mm-pairwiseuncertainty}
\chimp ~\equiv~ \sensormonitordef{\alpha}
\end{equation}

Monitor condition \eqref{eq:mm-pairwiseuncertainty} is satisfied if there exists a possible true value $y$ around measurement $\hat{y}$ and a possible next true value $y^+$ produced by the program $\alpha$ around the next measurement.
Now monitor condition \eqref{eq:mm-pairwiseuncertainty} tells us that there exist true values close to the measurements; in order for these true values to have the desired properties, 
the program $\alpha$ must be $[-\Delta,\Delta]$-contraction-safe, see \rref{thm:mm-pairwisecorrectness}.

\begin{theorem}[Pairwise measurement monitor preserves invariants]
\label{thm:mm-pairwisecorrectness}
Let program $\alpha$ in $\Delta$-measurement normal form be $[-\Delta,\Delta]$-contraction-safe.
Assume the system transitions from $\stold$ to $\stnew$, which agree on $\scomplement{\boundvars{\alpha}}$, with non-faulty sensors. 
If the contraction holds in the beginning ($\imodels{\I}{\marginProp{y}{\hat{y}}{\Delta}{\inv}}$) and the pairwise measurement monitor $\chimp$ with \eqref{eq:mm-pairwiseuncertainty} is satisfied $\irelmodels{\I}{\It}{\chimp}$,
then $\inv$ is preserved, \ie, $\imodels{\It}{\inv}$.
\end{theorem}

\begin{proof}
By contraction-safety, logical state relation, and coincidence (\rref{app:proofs}).
\end{proof}
\begin{proofatend}
The model monitor is satisfied $\irelmodels{\I}{\It}{\sensormonitordef{\alpha}}$ by assumption, so from $y^+ \not\in \vars{\alpha}$ we get $\irelmodels{\I}{\It}{\exists \nsimilar{y}{\hat{y}}{\Delta}\,\exists y^+\,\didia{\alpha}\Upsilon^+}$ by Barcan \cite{DBLP:journals/jar/Platzer17}.
Hence there exist $r\in\mathbb{R}$ for $y$ and $s\in\mathbb{R}$ for $y^+$ as well as states $\stold_y^r$ and $\stnew_y^s$ with $\imodels{\Iw[\stold_y^r]}{y{\in}\neighborhood{\hat{y}}{\Delta}}$ and $\irelmodels{\Iw[\stold_y^r]}{\Iw[\stnew_y^s]}{\didia{\alpha}\Upsilon^+}$ and so $\iaccessible[\alpha]{\Iw[\stold_y^r]}{\Iw[\stnew_y^s]}$ by \rref{lem:logicalrelation}.
We get $\imodels{\Iw[\stold_y^r]}{\marginProp{y}{\hat{y}}{\Delta}{\inv}}$ from $\imodels{\I}{\marginProp{y}{\hat{y}}{\Delta}{\inv}}$ by $\stold_y^r=\stold$ on $\scomplement{\{y\}}$ and $y \not\in \freevars{\marginProp{y}{\hat{y}}{\Delta}{\inv}}$ with \rref{lem:coincidence}.
Since program $\alpha$ is $[-\Delta,\Delta]$-contraction-safe by assumption, we now know $\imodels{\Iw[\stold_y^r]}{\dibox{\alpha}{\marginProp{y}{\hat{y}}{\Delta}{\inv}}}$ by \rref{def:contraction} using $\imodels{\Iw[\stold_y^r]}{\marginProp{y}{\hat{y}}{\Delta}{\inv}}$.
Hence $\imodels{\Iw[\stnew_y^s]}{\marginProp{y}{\hat{y}}{\Delta}{\inv}}$ by $\iaccessible[\alpha]{\Iw[\stold_y^r]}{\Iw[\stnew_y^s]}$.
Since $\stnew=\stnew_y^s$ on $\scomplement{\{y\}}$ we get\linebreak $\imodels{\It}{\marginProp{y}{\hat{y}}{\Delta}{\inv}}$ by \rref{lem:coincidence} since $y \not\in \freevars{\marginProp{y}{\hat{y}}{\Delta}{\inv}}$. 
Therefore from\linebreak $\imodels{\It}{y{\in}\neighborhood{\hat{y}}{\Delta}}$ by assumption we conclude $\imodels{\It}{\inv}$. 
\end{proofatend}

Besides safety, monitoring for existence of an unobservable value $y$ that fits to the present measurement $\hat{y}$ also guarantees bounded variation between true values:

\begin{proposition}[Bounded variation coincidence]
\label{prop:boundedvariationcoincidence}
Let $\alpha$ be a hybrid program in $\Delta$-measurement normal form.
Assume the system transitions through the sequence of states $\nu_0$, $\nu_1$, \ldots $\nu_n$, which agree on $\scomplement{\boundvars{\alpha}}$, such that
$\irelmodels{\Iw[\nu_{i-1}]}{\Iw[\nu_i]}{\chimp}$ with \eqref{eq:mm-pairwiseuncertainty} for all $1{\leq}i{\leq}n$.
Then there are $\mstateapprox{\omega_{i-1}}{\nu_{i-1}}{y}{\Delta}$, $\mstateapprox{\mu_i}{\nu_i}{y}{\Delta}$ such that $\iaccessible[\alpha]{\Iw[\omega_{i-1}]}{\Iw[\mu_i]}$.
\end{proposition}

\begin{proof}
By logical state relation, see \rref{app:proofs}.
\end{proof}
\begin{proofatend}
The proof follows the sketch below.
Note that the true evolution measured in $\nu_{i-1}\rightarrow\nu_i\rightarrow\nu_{i+1}$ is not necessarily a \emph{connected} path since the measurements allow jumps: the endpoint $\mu_i$ of the $\alpha$-run $\iaccessible[\alpha]{\Iw[\omega_{i-1}]}{\Iw[\mu_i]}$ might be different from the start of the next 
$\iaccessible[\alpha]{\Iw[\omega_i]}{\Iw[\mu_{i+1}]}$.

\includegraphics{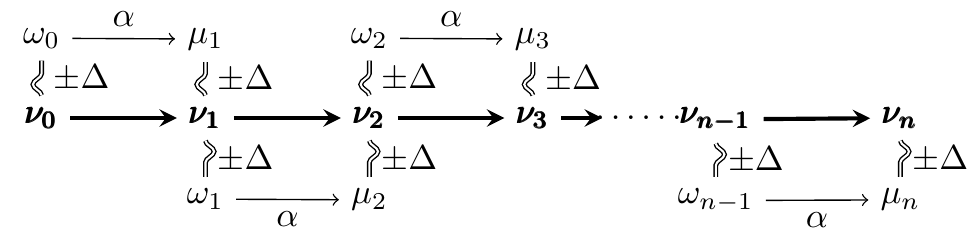}

\noindent 
From $\irelmodels{\Iw[\nu_{i-1}]}{\Iw[\nu_i]}{\chimp}$, \ie, $\irelmodels{\Iw[\nu_{i-1}]}{\Iw[\nu_i]}{\sensormonitordef{\alpha}}$ we get
$\irelmodels{\Iw[\nu_{i-1}]}{\Iw[\nu_i]}{\exists\nsimilar{y}{\hat{y}}{\Delta}\,\exists y^+\, \didia{\alpha}\Upsilon^+}$ by Barcan \cite{DBLP:journals/jar/Platzer17} with $y^+ \not\in \vars{\alpha}$.
Hence there exist $\mstateapprox{\omega_{i-1}}{\nu_{i-1}}{y}{\Delta}$ and $\mstateapprox{\mu_i}{\nu_i}{y}{\Delta}$ with $\irelmodels{\Iw[\omega_{i-1}]}{\Iw[\mu_i]}{\didia{\alpha}\Upsilon^+}$ and we conclude $\iaccessible[\alpha]{\Iw[\omega_{i-1}]}{\Iw[\mu_i]}$ by \rref{lem:logicalrelation}.
\end{proofatend}

The true evolution measured in $\nu_{i-1}\rightarrow\nu_i\rightarrow\nu_{i+1}$ is not necessarily a \emph{connected} path since the true values are freshly estimated from only the last measurement, which allows jumps: the true value $\mu_i(y)$ of the $\alpha$-run $\iaccessible[\alpha]{\Iw[\omega_{i-1}]}{\Iw[\mu_i]}$ might be different from the start $\omega_i(y)$ of the next 
$\iaccessible[\alpha]{\Iw[\omega_i]}{\Iw[\mu_{i+1}]}$, but the values are reasonably close $\nsimilar{\omega_i(y)}{\mu_i(y)}{2\Delta}$ by \rref{def:uncertaintysimilarity}.
\rref{prop:mm-singlestepdistance} and \rref{cor:mm-npairwiseuncertaintycorrectness} bound the variation in the possible true values $y$ of $\alpha$-steps when a monitor is satisfied over measured values $\hat{y}$.

\begin{proposition}[Single-step variation]
\label{prop:mm-singlestepdistance}
Let $\alpha$ be a program in $\Delta$-measurement normal form.
Let $z{-}z_0$ denote the interpolated effect of $\plantDur{\ctrlOut}{\varepsilon}$ in $\alpha$.
If $\irelmodels{\I}{\It}{\chimp}$ with \eqref{eq:mm-pairwiseuncertainty}, then the variation distance in a single $\alpha$-step is bounded: 
$\nsimilar{\stnew(y)}{\stold(y) + z{-}z_0}{{2\Delta}}$. 
\end{proposition}

\begin{proof}
By bounded variation coincidence and uncertainty similarity (\rref{app:proofs}).
\end{proof}
\begin{proofatend}
From $\irelmodels{\I}{\It}{\chimp}$ and non-faulty measurement we know there exist $\mstateapprox{\tilde{\stold}}{\stold}{y}{\Delta}$ and $\mstateapprox{\tilde{\stnew}}{\stnew}{y}{\Delta}$ with $\iaccessible[\alpha]{\tI}{\tIt}$ by \rref{prop:boundedvariationcoincidence} and so there are intermediate states $\stintermediate$ and $\tilde{\stintermediate}$ such that $\iaccessible[\senseCtrl]{\tI}{\Ii}$, $\iaccessible[\plantDur{\ctrlOut}{\varepsilon}]{\Ii}{\tIi}$, and $\iaccessible[\measure{y}{\Delta}]{\tIi}{\tIt}$.
Thus, the interpolated effect $z-z_0$ of $\plantDur{\ctrlOut}{\varepsilon}$ is $\tilde{\stintermediate}(y)-\stintermediate(y)$.
Since $y \not\in \boundvars{\senseCtrl}$ and $y \not\in \boundvars{\measure{y}{\Delta}}$ we get $\stintermediate(y)=\tilde{\stold}(y)$ and $\tilde{\stnew}(y)=\tilde{\stintermediate}(y)$.
Hence, $\tilde{\stnew}(y) = \tilde{\stold}(y) + z-z_0$ and in turn
$\nsimilar{\tilde{\stnew}(y)}{\stold(y) + z-z_0}{\Delta}$
from $\mstateapprox{\tilde{\stold}}{\stold}{y}{\Delta}$. 
We conclude $\nsimilar{\stnew(y)}{\stold(y) + z-z_0}{2\Delta}$ from $\mstateapprox{\tilde{\stnew}}{\stnew}{y}{\Delta}$.
\end{proofatend}

\newcommand{\multistepvariation}{\ensuremath{2\Delta(n + 1)}}

\begin{corollary}[Multi-step variation]
\label{cor:mm-npairwiseuncertaintycorrectness}
Assume the model transitions through a sequence of states $\nu_0$, $\nu_1$, \ldots $\nu_n$ with interpolated plant effects $z_i{-}z_{i-1}$.
If $\irelmodels{\Iw[\nu_{i-1}]}{\Iw[\nu_i]}{\chimp}$ with \eqref{eq:mm-pairwiseuncertainty} for all $1{\leq}i{\leq}n$ then variation is bounded: $\nsimilar{\nu_n(y)}{\nu_0(y) + \sum_{i=1}^n(z_i{-}z_{i-1})}{{\multistepvariation}}$.
\end{corollary}

\begin{proof}
By single-step variation of consecutive measurements, see \rref{app:proofs}.
\end{proof}
\begin{proofatend}
Follows from $2\Delta$ variation of consecutive measurements, \ie, $\nsimilar{\hat{y}^+}{\hat{y}}{{2\Delta}}$ and the additional sensor uncertainty $\nsimilar{y_i}{\hat{y}_i}{\Delta}$ that is relevant at the beginning and end of the sequence of states but not at intermediate steps since overlapping.
\end{proofatend}

\newcommand{\quantifiedsensorviolationdef}[1]{\ensuremath{\forall \nsimilar{y}{\hat{y}}{\Delta}\,\dbox{#1}\left(\forall y^+\,\neg \Upsilon^+\right)}}
\newcommand{\sensorviolationdef}[1]{\ensuremath{\forall \nsimilar{y}{\hat{y}}{\Delta}\,\dbox{#1}\left(\Upsilon_{\boundvars{\ctrl{x};\plant{\ctrlOut}}\backslash y} \limply y\not\in\neighborhood{\hat{y}^+}{\Delta}\right)}}
\newcommand{\metricsensorviolationdef}[1]{\ensuremath{\forall \nsimilar{y}{\hat{y}}{\Delta}\,\dbox{#1}\forall y^+\,(\nsimilar{y^+}{y}{\violationbound})}}

As a consequence of \rref{prop:mm-singlestepdistance} and \rref{cor:mm-npairwiseuncertaintycorrectness}, \rref{thm:mm-pairwisecorrectness} is useful for single-step consistency checks: such a monitor 
\begin{inparaenum}[(i)]
\item keeps at least $\Delta \geq 0$ safety margin with a contraction-safe controller because otherwise freshly estimating true $y$ from measurements on every iteration does not preserve invariants,  and 
\item detects ``large'' deviations that occur in a single monitoring step. 
\end{inparaenum}
However, pairwise consistency has to be safeguarded for its entire $\multistepvariation$ deviation over $n$ steps to ensure safety despite gradual drift, as illustrated in the following example, which inhibits motion and cannot exploit improving the safety margin $\Delta$ over a history of measurements.

\newcommand{\vUncertainty}{\ensuremath{\Delta}}

\begin{mdframed}[style=example]
\begin{example}
The flight protocol so far assumes perfect knowledge about the intruder's linear velocity $\vintruder$.
Here, we extend the protocol when the ownship takes measurements $\hat{\vintruder}$ of the intruder velocity $\vintruder$ with a sensor that might be off by uncertainty $\vUncertainty$.
\vspace{-.5\baselineskip}
\[
\alpha \equiv \bigl(\pchoice{(\humod{\wownship}{0};\ptest{\forall \nsimilar{\vintruder}{\hat{\vintruder}}{{\vUncertainty}}\, \straightinv})}{(\humod{\wownship}{1};\ptest{\forall \nsimilar{\vintruder}{\hat{\vintruder}}{{\vUncertainty}}\,\flightinv{\wownship}})}\bigr);~\plantPrg;~\measure{\vintruder}{{\vUncertainty}}
\]
Now the true linear intruder velocity $\vintruder$ is unobservable, so the monitor condition $\exists \nsimilar{\vintruder}{\hat{\vintruder}}{{\vUncertainty}}\, \didia{\alpha}\exists \vintruder^+(\flightUpsilon \land \hat{\vintruder}^+{=}\hat{\vintruder})$ extends $\Upsilon^+$ with $\hat{\vintruder}^+{=}\hat{\vintruder}$ and existentially quantifies away both unobservable speeds $\vintruder$ and $\vintruder^+$.
Because intruder speed in the model is constant, all
measurements $\hat{\vintruder}$ and $\hat{\vintruder}^+$ are at most $2\vUncertainty$ apart.
\rref{fig:pairwiseillustration} illustrates the monitor behavior when measuring the constant intruder speed $\vintruder=0.5$ with uncertainty $\vUncertainty=0.1$: the thick error bars represent the range of possible measurements $\hat{\vintruder}$ according to the true $\vintruder$, the thin error bars the measurement range allowed per $\nsimilar{\hat{\vintruder}^+}{\hat{\vintruder}}{{2\vUncertainty}}$ from the previous measurement.
The measurements $\hat{\vintruder}$ at $t=0$ to $t=5$ vary around the true values as would be expected from a sensor, which includes the worst case of two consecutive measurements hitting opposite bounds of the uncertainty.  
Since the monitor does not keep history, the true $\vintruder$ is allowed to drift at times $t=6$ and $t=7$.
The monitor detects violations if the $2\vUncertainty$ uncertainty is exceeded in a single step at $t=8$.
\end{example}
\end{mdframed}

\begin{figure}[tb]
\begin{subfigure}[b]{.48\linewidth}
\centering
\begin{footnotesize}
\begin{tikzpicture}
\begin{axis}[
  compat=newest,x dir=normal,
  legend pos=outer north east,
  width=\columnwidth,height=5cm,
  xlabel={t},ylabel={$\vintruder$},
  yticklabels=\empty,
  axis lines=middle,
  axis line style={->},
  x label style={at={(current axis.right of origin)},anchor=west},
  y label style={at={(current axis.above origin)},anchor=east},
  enlargelimits=true,
  clip=false
]
\addplot [color=lslightblue, thick, mark=square*]
 plot [error bars/.cd, y dir = both, y explicit, error bar style={line width=2pt}]
 table[x =t, y =v, y error =ds]{curvebotmodelmonitorpairwise.dat};
\addplot [color=lsblue, only marks, mark=*
]
 plot 
 table[x =t, y =snow]{curvebotmodelmonitorpairwise.dat};
\addplot [color=lsblue, thick, only marks, mark=off,mark options={scale=2}]
 plot [error bars/.cd, y dir = both, y explicit,error mark options={rotate=90,lslightblue,mark size=5pt}]
 table[x =t, y =spre, y error =dsp]{curvebotmodelmonitorpairwise.dat}; 
\node[anchor=east] at (axis cs: 9,1) {undetected excess drift at $t{=}7$}; 
\node[anchor=east,fill=white,xshift=-.5em] at (axis cs: 8,.1) {detected violation at $t{=}8$}; 

\addplot+[name path=trueupper,mark=none,dashed,smooth,very thick,color=lsblue,forget plot] table[x =t, y expr = {\thisrow{v} + \thisrow{ds}}]{curvebotmodelmonitorpairwise.dat}; 
\addplot+[name path=truelower,mark=none,dashed,smooth,very thick,color=lsblue,forget plot] table[x =t, y expr = {\thisrow{v} - \thisrow{ds}}]{curvebotmodelmonitorpairwise.dat}; 
\addplot+[name path=measureupper,mark=none,dotted,smooth,color=lsblue,forget plot] table[x =t, y expr = {\thisrow{spre} + \thisrow{dsp}}]{curvebotmodelmonitorpairwise.dat}; 
\addplot+[name path=measurelower,mark=none,dotted,smooth,color=lsblue,forget plot] table[x =t, y expr = {\thisrow{spre} - \thisrow{dsp}}]{curvebotmodelmonitorpairwise.dat}; 
\addplot[lsochre!35,opacity=0.5,
] fill between[of=measurelower and measureupper];

\addlegendentry{True speed $s$};
\addlegendentry{Measured speed $\hat{s}$};
\legend{};
\end{axis}
\end{tikzpicture}
\end{footnotesize}
\caption{Bounds from current measurement may accept excess drift}
\label{fig:pairwiseillustration}
\end{subfigure}
\hfill
\begin{subfigure}[b]{.48\linewidth}
\centering
\begin{footnotesize}
\begin{tikzpicture}
\begin{axis}[
  compat=newest,x dir=normal,
  legend pos=outer north east,
  width=\columnwidth,height=5cm,
  xlabel={t},ylabel={$\vintruder$},
  yticklabels=\empty,
  axis lines=middle,
  axis line style={->},
  x label style={at={(current axis.right of origin)},anchor=west},
  y label style={at={(current axis.above origin)},anchor=east},
  enlargelimits=true,
  clip=false
]
\addplot [color=lslightblue, thick, mark=square*]
 plot[error bars/.cd, y dir = both, y explicit, error bar style={line width=2pt}]
 table[x=t, y=v, y error expr={0.1}]{modelmonitorestimator.dat};
\addplot [color=lsblue, thick, only marks, mark=*
]
 plot
 table[x =t, y =s, y error minus=v-l, y error plus=v+u]{modelmonitorestimator.dat};
\addplot [color=lsblue, thick, only marks, mark=off,mark options={scale=2}]
 plot[error bars/.cd, y dir = both, y explicit] 
 table[x =t, y =sp, y error minus=sp-l, y error plus=sp+u]{modelmonitorestimator.dat};  
\node[anchor=east] at (axis cs: 8.5,1) {detected excess drift at $t{=}7$}; 
\node[anchor=east,fill=white,xshift=-.5em] at (axis cs: 8,.1) {detected violation at $t{=}8$}; 

\addplot+[name path=trueupper,mark=none,dashed,smooth,color=lsblue,forget plot] table[x =t, y expr = {\thisrow{s} + \thisrow{v+u}}]{modelmonitorestimator.dat}; 
\addplot+[name path=truelower,mark=none,dashed,smooth,color=lsblue,forget plot] table[x =t, y expr = {\thisrow{s} - \thisrow{v-l}}]{modelmonitorestimator.dat}; 
\addplot+[name path=measureupper,mark=none,dashed,smooth,very thick,color=lsblue,forget plot] table[x =t, y expr = {\thisrow{v} + 0.1}]{modelmonitorestimator.dat}; 
\addplot+[name path=measurelower,mark=none,dashed,smooth,very thick,color=lsblue,forget plot] table[x =t, y expr = {\thisrow{v} - 0.1}]{modelmonitorestimator.dat}; 
\addplot+[name path=estimateupper,mark=none,dotted,smooth,color=lsblue,forget plot] table[x =t, y expr = {\thisrow{sp} + \thisrow{sp+u}}]{modelmonitorestimator.dat}; 
\addplot+[name path=estimatelower,mark=none,dotted,smooth,color=lsblue,forget plot] table[x =t, y expr = {\thisrow{sp} - \thisrow{sp-l}}]{modelmonitorestimator.dat}; 
\addplot[lsgreen] fill between[of=truelower and trueupper];
\addplot[lsochre!35,opacity=0.5
] fill between[of=estimatelower and estimateupper];

\addlegendentry{True speed $s$};
\addlegendentry{Measured speed $\hat{s}$};

\legend{};
\end{axis}
\end{tikzpicture}
\end{footnotesize}
\caption{Bounds from entire measurement history detects drift}
\label{fig:historyillustration}
\end{subfigure}
\caption{Estimating bounds from measurement pairs and measurement history}.
\end{figure}
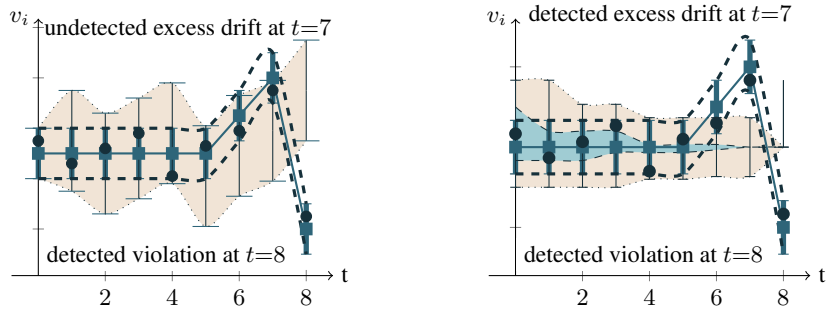

Monitoring contractions and single-step conformance guarantees safety up to the current measurement per \rref{thm:mm-pairwisecorrectness} (beyond by Theorems \ref{thm:controlviolationrecoverability}+\ref{thm:modelviolationrecovarability}) and detects large single-step deviations per \rref{prop:boundedvariationcoincidence} (\eg, at $t=8$ in \rref{fig:pairwiseillustration}), but cannot detect gradual drift early.
In order to react to drift in measurements even before contractions are violated, we extend our monitors with state estimation over the entire measurement history.

\subsection{Model Monitors for Rolling Consistency of Measurements}

Even if individual measurement pairs do not trigger the detectable single-step violation of the previous section, the resulting aggregated drift over multiple measurements is still detectable when we keep a history of the control choices.
Instead of an explicit list of measurement and control histories, we only represent in aggregate form what really matters: acceptable bounds for the upcoming true values (solid blue small range in \rref{fig:historyillustration}) and measurements (solid light-brown large range in \rref{fig:historyillustration}).
The bounds are updated on each monitor execution with the current control choice.
The resulting \emph{rolling state estimator} detects gradual violation over the course of multiple measurements.

\begin{definition}[Non-diverging rolling state estimator]
\label{def:nondivergingestimator}
A \emph{rolling state estimator} \(\pumod{[l,u]}{e(\hat{y}_0,\hat{y},y-y_0,\Delta},[l_0,u_0])\) updates the estimate $[l,u]$ from the previous measurement $\hat{y}_0$, current measurement $\hat{y}$, the modeled interpolated plant effect $y{-}y_0$ and the previous estimate $[l_0,u_0]$ with $\nsimilar{y_0}{\hat{y}_0}{[l_0,u_0]}$ and $[l_0,u_0] \subseteq [-\Delta,\Delta]$.
The estimator \(\pumod{[l,u]}{e(\hat{y}_0,\hat{y},y-y_0,\Delta},[l_0,u_0])\)
is \emph{non-diverging} if $u-l \leq u_0-l_0$ for all $\hat{y}_0,\hat{y},y_0,y,\Delta,l,u$ with $\nsimilar{y}{\hat{y}}{[l,u]}$ and $\neighborhood{\hat{y}}{[l,u]} \subseteq \neighborhood{\hat{y}}{\Delta}$.
\end{definition}

The rolling state estimator updates estimates $\nsimilar{y}{\hat{y}}{[l,u]}$ of true $y$ on every measurement such that the history of measurements is preserved in aggregate form.
On each step, the monitor checks for existence of a true state $y$ in the estimate, and uses the rolling state estimator to incorporate the current measurement $\hat{y}$ into the estimate for the next check.
That way, the interpolated plant effect $y{-}y_0$ is reflected in the measurement bounds in the monitor condition \eqref{eq:mm-rollingstateestimator} below, so does not need to be observable.
For a sequence of $n$ measurements, a non-diverging rolling state estimator keeps tighter bounds compared to the $\multistepvariation$ bounds of \rref{cor:mm-npairwiseuncertaintycorrectness} without measurement history.
Note that measurements typically vary due to sensor uncertainty, so the estimate almost surely even improves over time by observing more measurements.

For a hybrid program $\alpha$ in $\Delta$-measurement normal form, the monitor condition \eqref{eq:mm-rollingstateestimator} checks the plausibility of a history of measurements by estimating the true $\nsimilar{y}{\hat{y}}{{[l,u]}}$ from the observations. 
\newcommand{\rollingestimatormonitor}[1]{\exists \nsimilar{y}{\hat{y}}{{[l,u]}}\,\didia{#1}\left(\exists y^+\,\Upsilon^+\right)}%
\newcommand{\longrollingestimatormonitor}{\rollingestimatormonitor{\overbrace{\recallState{y}}^{\text{remember previous state}};~ \alpha;\\ \underbrace{\updateEstimator{y}{\Delta}{l}{u}}_{\text{update estimator}}}}%
\vspace{-\baselineskip}
\begin{multline}\label{eq:mm-rollingstateestimator}
\chimp ~\equiv~ \longrollingestimatormonitor
\end{multline}

The history of $\alpha$ is reflected by the rolling state estimator in its $[l,u]$-bounds for upcoming measurements, which guarantees safety by \rref{thm:mm-rollingstateestimatorcorrectness} with Theorems \ref{thm:controlviolationrecoverability}+\ref{thm:modelviolationrecovarability}.

\begin{theorem}[Monitor with rolling state estimator maintains invariants]
\label{thm:mm-rollingstateestimatorcorrectness}
Let $\alpha$ be a $[l,u]$-contraction safe hybrid program in $\Delta$-measurement normal form and $\estimatorname$ be a non-diverging estimator.
Assume the system transitions from $\stold$ to $\stnew$, which agree on $\scomplement{\boundvars{\alpha}}$, with non-faulty sensors. 
If the $[l,u]$-contraction holds in the beginning $\imodels{\I}{\marginProp{y}{\hat{y}}{{[l,u]}}{\inv}}$ and the monitor condition $\chimp$ with \eqref{eq:mm-rollingstateestimator} is satisfied, \ie, $\irelmodels{\I}{\It}{\chimp}$,
then the invariant $\inv$ is maintained: $\imodels{\It}{\inv}$.
\end{theorem}

\begin{proof}
See \rref{app:proofs}.
\end{proof}
\begin{proofatend}
For a hybrid program $\alpha$ in $\Delta$-measurement normal form, let $\beta$ abbreviate the program $\recallState{y};~ \alpha;\updateEstimator{y}{\Delta}{l}{u}$.
The monitor condition $\chimp$ is satisfied $\irelmodels{\I}{\It}{\rollingestimatormonitor{\beta}}$ by assumption, so we get $\irelmodels{\I}{\It}{\exists \nsimilar{y}{\hat{y}}{{[l,u]}}\,\exists y^+\,\didia{\beta}\Upsilon^+}$ by Barcan \cite{DBLP:journals/jar/Platzer17} with $y^+ \not\in \vars{\beta}$.
Hence there exist $r\in\mathbb{R}$ for $y$ and $s\in\mathbb{R}$ for $y^+$ as well as states $\stold_y^r$ and $\stnew_y^s$ with $\imodels{\Iw[\stold_y^r]}{y{\in}\neighborhood{\hat{y}}{[l,u]}}$ and $\irelmodels{\Iw[\stold_y^r]}{\Iw[\stnew_y^s]}{\didia{\alpha}\Upsilon^+}$ and so $\iaccessible[\beta]{\Iw[\stold_y^r]}{\Iw[\stnew_y^s]}$ by \rref{lem:logicalrelation}.
We get $\imodels{\Iw[\stold_y^r]}{\marginProp{y}{\hat{y}}{[l,u]}{\inv}}$ from assumption $\imodels{\I}{\marginProp{y}{\hat{y}}{[l,u]}{\inv}}$ by \rref{lem:coincidence} since $\stold_y^r=\stold$ on $\scomplement{\{y\}}$ and $y \not\in \freevars{\marginProp{y}{\hat{y}}{[l,u]}{\inv}}$.
Program $\alpha$ is $[l,u]$-contraction-safe by assumption and therefore by \rref{def:contraction} using $\imodels{\Iw[\stold_y^r]}{\marginProp{y}{\hat{y}}{[l,u]}{\inv}}$ we get
\[\imodels{\Iw[\omega_y^r]}{\dibox{\recallState{y};\senseCtrl;\plantDur{\ctrlOut}{\varepsilon};\measure{y}{\Delta}}{\marginProp{y}{\hat{y}}{[l,u]}{\inv}}} \enspace .\]
Now $\imodels{\Iz}{\marginProp{y}{\hat{y}}{[l,u]}{\inv}}$ for all reachable states $\iget[state]{\Iz}$ such that \[\iaccessible[\recallState{y};\senseCtrl;\plantDur{\ctrlOut}{\varepsilon};\measure{y}{\Delta}]{\Iw[\stold_y^r]}{\Iz}\enspace.\] 
In particular, since the monitor condition is satisfied there exists a state $\iget[state]{\Iz}$ such that also $\iaccessible[\updateEstimator{y}{\Delta}{l}{u}]{\Iz}{\Iw[\stnew_y^s]}$.
The estimator $\updateEstimator{y}{\Delta}{l}{u}$ is non-diverging (\rref{def:nondivergingestimator}) by assumption, so we get $\imodels{\Iw[\stnew_y^s]}{\marginProp{y}{\hat{y}}{{[l,u]}}{\inv}}$ and in turn also $\imodels{\It}{\marginProp{y}{\hat{y}}{{[l,u]}}{\inv}}$ by \rref{lem:coincidence} with $\stnew=\stnew_y^s$ on $\scomplement{\{y\}}$ and $y \not\in \freevars{\marginProp{y}{\hat{y}}{{[l,u]}}{\inv}}$. 
Hence, from $\imodels{\It}{y{\in}\neighborhood{\hat{y}}{{[l,u]}}}$ by assumption we conclude $\imodels{\It}{\inv}$. 
\end{proofatend}

\begin{mdframed}[style=example]
\begin{example}
We extend the flight protocol with a velocity estimator $e(\hat{\vintruder}_0,\hat{\vintruder},0,[l_0,u_0])$ that updates lower and upper bounds on the deviation between the next measurement $\hat{\vintruder}$ and the true velocity $\vintruder$ from the current measurement $\hat{\vintruder}_0$ and the current estimation bounds $[l_0,u_0]$.
Since $\vintruder$ is constant (but not perfectly known), the interpolated plant effect $\vintruder-{\vintruder}_0$ is $0$:
\(
l=\max{\bigl({-}\Delta,\hat{\vintruder}_0-\hat{\vintruder}+l_0\bigr)}\qquad
u=\min{\bigl(\Delta,\hat{\vintruder}_0-\hat{\vintruder}+u_0\bigr)}
\)

\end{example}
\end{mdframed}

\section{Implementation and Evaluation}
\label{sec:evaluation}

Based on the monitor characterizations developed here, the process for synthesizing model monitors from hybrid systems models is systematic and correct by construction \cite{DBLP:journals/fmsd/MitschP16}, implemented as a synthesis tactic in \KeYmaeraX~\cite{DBLP:conf/cade/FultonMQVP15} see \rref{app:implementation}.
A crucial additional step in the process is to eliminate the remaining (existential) quantifiers that describe possible unobservable true values, but the complexity and duration hinges on the performance of external solvers.
Additional arithmetical simplifications beyond \cite[Opt. 1]{DBLP:journals/fmsd/MitschP16} help overcome the limitations \cite{Davenport1988} of quantifier elimination procedures.

\paragraph{Quantifier Elimination Preprocessing.}

In order to make quantifier elimination tractable even in complex models, we exploit two important observations about the typical shape of input programs that is consequently reflected in the synthesized monitors: 
\begin{inparaenum}[(i)]
\item controllers have different control choices over multiple control branches, which result in alternative paths through the program, reflected as disjunctions in the monitor, and 
\item controllers do not mention unobservable variables (only the measured quantities), hence unobservable variables occur only in few subformulas of a monitor.
\end{inparaenum}
We rewrite monitor conditions into disjunctive normal form by proof to split a single quantifier elimination proof obligation into several smaller ones with the lemmas \eqref{eq:existsOr} and \eqref{eq:existsVacuous} corresponding to these observations.
These preprocessing steps in the tactic before quantifier elimination help scale the synthesis tactic to larger models, as discussed next.
\begin{align}
\exists x\, \bigl(p(x) \lor q(x)\bigr) & \lbisubjunct \exists x\, p(x) \lor \exists x\, q(x) \label{eq:existsOr}\\
\exists x\, \bigl(p() \land q(x)\bigr) & \lbisubjunct p() \land \exists x\, q(x) \label{eq:existsVacuous}
\end{align}

\paragraph{Synthesis Performance.}

We ran the synthesis tactic on the running example and \dL case studies, a water tank \cite{DBLP:journals/fmsd/MitschP16}, train control \cite{DBLP:conf/icfem/PlatzerQ09}, road traffic control \cite{DBLP:conf/iccps/MitschLP12}, and robot collision avoidance \cite{DBLP:journals/ijrr/MitschGVP17}. 
The steps of the synthesis process, its duration, and the size of the resulting monitor condition in terms of operators are summarized in \rref{tab:evaluation}.
The duration measurements were taken on a 2.4 GHz Intel Core i7 with 16GB of memory.

\begin{table}[tb]
\caption{Monitor synthesis case studies}
\label{tab:evaluation}
\begin{footnotesize}
\begin{tabularx}{\columnwidth}{
	X@{\hskip 1ex}
	l@{\hskip -1ex}
	r@{\hskip 1ex}
	r@{\hskip 2ex}
	r@{\hskip 1ex}
	r@{\hskip 1ex}
	r@{\hskip 1ex}
	r@{\hskip 1ex}
	r@{\hskip 1ex}
	r@{\hskip 1ex}
	r@{\hskip 1ex}
	r@{\hskip 1ex}
}
\toprule    
    &&
    & \multicolumn{2}{@{}c@{}}{Monitor Size}
    &
    & \multicolumn{4}{c}{Synthesis Duration [s]}
    & \multicolumn{2}{c}{Simulation}
\tabularnewline
\cmidrule{4-5}
\cmidrule(r){7-10}
\cmidrule(l){11-12}
\multicolumn{2}{l}{Case Study} & Dim. & $\forall\exists$ & $\forall\exists$-free & \parbox{2em}{\raggedleft Proof Steps} & \parbox{3em}{\raggedleft Proof Check} & \parbox{3em}{\raggedleft Dis\-co\-very} & QE & Ext. & P & R
\tabularnewline
\midrule
\multirow{3}{*}{\parbox{1.5cm}{\raggedright Horizontal flight}}
& original     & 8 & --  & 46 & 6259  & 33  & 4 & --& 1 & 1 & 1 \\
& + actuator   & 9 & 124 & 43 & 10369 & 33 & 6 & 3 & 4 & 1 & 1 \\
& + sensor     & 9 & 102 & 169 & 11804 & 12  & 8 & 2 & 3 & 1 & 0.81$^\text{a}$
\tabularnewline
\midrule
\multirow{3}{*}{\parbox{1.5cm}{\raggedright Water tank~\cite{DBLP:journals/fmsd/MitschP16}}}
& original     & 3 & -- & 19  & 676  & 2 & 1 & -- & 0 & 1 & 1 \\
& + actuator   & 4 & 25 & 30  & 1864 & 2  & 2 & 1 & 1 & 1 & 1 \\
& + sensor     & 4 & 35 & 100 & 1383 & 2  & 2 & 0 & 0 & 0.94 & 0.93
\tabularnewline
\midrule
\multirow{3}{*}{\parbox{1.5cm}{\raggedright Train control~\cite{DBLP:conf/icfem/PlatzerQ09}}}
& original     & 8  & --  & 74  & 1309 & 6 & 3 & -- & 0 & 1 & 1 \\
& + actuator   & 10 & 148 & 109  & 3703 & 8 & 6 & 1 & 2 & 1 & 1 \\
& + sensor     & 10 & 137 & 206 & 3051 & 6 & 5 & 0 & 0 & 1 & 0.98
\tabularnewline
\midrule
\multirow{3}{*}{\parbox{1.5cm}{\raggedright Road traffic control~\cite{DBLP:conf/iccps/MitschLP12}}}
& original     & 9 & --  & 197 & 11692 & 146 & 17 & -- & 104 & \multicolumn{2}{c}{\multirow{3}{*}{\parbox{1.5cm}{\centering Not\\ simulated}}} \\
& + actuator   & 11 & 218 & 803 & 12551 & 193 & 16 & 0 & 179 \\
& + sensor     & 11 & 208 & 877 & 22384 &  83 & 25 & 0 & 69
\tabularnewline
\midrule
\multirow{3}{*}{\parbox{1.55cm}{Robot collision avoidance~\cite{DBLP:journals/ijrr/MitschGVP17}}}
& original     &  & --  & 231 & 26405 & 332 & 61 & -- & 16 & \multicolumn{2}{c}{\multirow{3}{*}{\parbox{1.5cm}{\centering Not\\ simulated}}} \\
& + actuator   & 15 & 535 & 275 & 4786 &  668 & 95 & 36$^\text{b}$ &  38\\
& + sensor     & 15 & 513 & 10980$^\text{c}$ & 152810 & 675 & 89 & 1142$^\text{b}$ & 1145
\tabularnewline
\bottomrule
\end{tabularx}
$^\text{a}$ 5 runs with 15 loop iterations \quad $^\text{b}$ with preprocessing \eqref{eq:existsOr} and \eqref{eq:existsVacuous} \quad $^\text{c}$ simplification aborted (2min timeout)
\end{footnotesize}
\end{table}

For each case study, we synthesized monitors for the original model and analyzed extended models that include sensor uncertainty and actuator disturbance.
The column ``Dim.'' gives an intuition on the complexity of the case study in terms of the number of model variables.
We list the monitor size in terms of the number of arithmetical, comparison, and logical operators in the intermediate quantified form (column ``$\forall\exists$'') and the final quantifier-free fully simplified form (column ``$\forall\exists$-free''), as well as the duration of the synthesis steps: column ``Proof Check'' lists the duration of checking the safety proof, column ``Discovery'' lists the duration of discovering the intermediate quantified monitor form, and column ``QE'' the duration of obtaining the quantifier-free form.
Finally, column ``Ext.'' lists how much of the total synthesis duration (Proof Check + Discovery + QE) is spent in external solvers.
The main insight is that the synthesis and discovery of monitor conditions with the techniques in this paper processes CPS models with modest computation and time resources.

\vspace{-\baselineskip}
\paragraph{Monitor Performance.}

We ran simulations based on hybrid program image computations \cite{DBLP:conf/hybrid/PlatzerC07} with randomly injected actuator disturbance recorded as ground truth, and measured the monitoring outcome in terms of precision ($\frac{\text{true non-alarms}}{\text{all non-alarms}}$, column ``P'') and recall ($\frac{\text{true non-alarms}}{\text{true non-alarms}+\text{false alarms}}$, column ``R''), see \rref{tab:evaluation}.
Simulation runs executed 50 loop iterations, and the measurements were averaged over 100 runs.
As expected, all simulation steps in lines ``original'' and ``actuator'' with full observability (\ie, sensors work perfectly) are correctly classified by the monitors.
Simulation steps in lines ``sensor'' randomly chose sensor uncertainty i.i.d.\ on each step.

\section{Related Work}\label{sec:relatedwork}

Runtime verification and monitoring for finite state discrete systems has received significant attention (\eg, \cite{DBLP:conf/time/DAngeloSSRFSMM05,DBLP:journals/sttt/HavelundR04,DBLP:journals/sttt/MeredithJGCR12}).
Some approaches monitor continuous-time signals (\eg, \cite{DBLP:conf/cav/DonzeFM13,DBLP:conf/formats/NickovicM07}).
We focus on hybrid systems models of CPS to combine both, and our methods are robust to the crucial effects of sensor uncertainty and actuator disturbance.

Several tools for formal verification of hybrid systems are actively developed (\eg, SpaceEx \cite{DBLP:conf/cav/FrehseGDCRLRGDM11}, dReach \cite{DBLP:conf/tacas/KongGCC15}, and extended NuSMV/MathSat \cite{DBLP:conf/cav/CimattiMT11}). 
Provably correct monitor synthesis, however, crucially relies on the rewriting capabilities and flexibility of combining $\dbox{\cdot}$ and $\didia{\cdot}$ modalities in \dL~\cite{Platzer18,DBLP:journals/jar/Platzer17} and \KeYmaeraX~\cite{DBLP:conf/cade/FultonMQVP15}.

\vspace{-.5\baselineskip}
\paragraph{Combined Offline and Runtime Verification.}

In \cite{DBLP:conf/rv/DesaiDS17,DBLP:journals/corr/abs-1808-07921}, offline model checking is combined with runtime monitoring for robot path planning.
For offline verification, the methods \emph{assume} that motion of the robot stays inside a tube around the planned path; staying inside the tube is monitored at runtime and enforced with fallback control synthesized using the techniques in \cite{DBLP:conf/cdc/HerbertCHBFT17}.
The approach uses STL to model motion primitives for runtime monitoring and learns parameters from example trajectories.
We, in contrast, use hybrid systems models including sensor uncertainty and actuator perturbation and check physical model compliance instead of assuming it when characterizing monitor conditions and fallback requirements with provable safety guarantees.

Reachset conformance testing \cite{DBLP:conf/hybrid/RoehmOWA16} computes reachable sets of hybrid automata at runtime to to falsify simulations or recorded data. 
The crucial benefit of our methods is to \emph{perform expensive computations offline} (provably correct fast reachable set computation online in realtime is hard) and \emph{provably guarantee safety} from offline proofs when the monitor conditions are satisfied at runtime of the monitored system \emph{from sensor measurements} and control decisions that are subject to \emph{actuator disturbance}.

\vspace{-.5\baselineskip}
\paragraph{Monitoring and Sandboxing.}
Owing to their practical significance, numerous monitoring and sandboxing techniques have been proposed for CPSs \cite{DBLP:conf/tacas/LiSSS14,DBLP:journals/fmsd/KonighoferABHKT17,DBLP:journals/acta/BloemCGHHJKK14,DBLP:conf/ecrts/KimVBKLS99,DBLP:conf/hybrid/SankaranarayananF12,DBLP:conf/rv/DokhanchiHF14,DBLP:journals/fmsd/DeshmukhDGJJS17}.
Because these approaches, with the notable exception of ModelPlex \cite{DBLP:journals/fmsd/MitschP16} that this work is based on, do not ship their monitors with correctness proofs, they may not check all conditions that are needed to discover all model violations and so no guarantee can be given that they always reliably engage fallback mitigation when necessary.

Specification mining techniques for LTL can be adapted to monitor for safety violations \cite{DBLP:conf/tacas/LiSSS14} and intervene, \emph{assuming} that the \emph{next environment input} is available to the monitor.
In CPSs, this is feasible only when the next input can be prevented from becoming actuated, see \rref{thm:controlviolationrecoverability} and, as presented in this paper, with means to detect gradual deviation from the model that accumulates to violation over time.

Shields \cite{DBLP:journals/fmsd/KonighoferABHKT17} and robust reactive system synthesis \cite{DBLP:journals/acta/BloemCGHHJKK14} are approaches to detect and correct erroneous control output in discrete models, which however ignore the continuous behavior that is crucial in CPS and explicit in our differential equations.
Monitoring based on discrete models is useful for high-level planning tasks (\eg, waypoint planning), but gives no guarantees about the resulting continuous physical motion and is unable to detect effects related to disturbance or sensing, such as gradual sensor drift.

Languages for modeling runtime monitors based on sensor events \cite{DBLP:conf/ecrts/KimVBKLS99} are purely discrete (e.g., speed lower than threshold), come without correctness guarantees on the mapping between monitor and inputs/outputs and without correctness guarantees on the safety properties and alarms.
In contrast, our methods \emph{provably guarantee} that satisfied monitors at runtime imply system safety (and in particular safety of the resulting physical effects) by relating the observed dynamics to the safe models verified offline.

Robustness estimation methods \cite{DBLP:conf/hybrid/SankaranarayananF12,DBLP:conf/rv/DokhanchiHF14,DBLP:journals/fmsd/DeshmukhDGJJS17} measure the degree to which a monitor given as a signal/metric temporal logic specification is satisfied in order to allow bounded perturbation akin to our actuator disturbance, but cannot detect gradual drift in sensor measurements. 
The methods assume a finite time horizon, compact inputs and outputs, and restrictions on the dynamics (\eg, piecewise constant between sampling points \cite{DBLP:journals/fmsd/DeshmukhDGJJS17}), but do not support the predictive model of continuous dynamics and sensors/actuators that is needed for system safety at runtime, which we handle explicitly.

\vspace{-.5\baselineskip}
\paragraph{Summary.}

In summary, our approach improves over existing runtime monitoring techniques with provably correct monitor conditions, explicit dynamics with sensor uncertainty and actuator disturbance, and shifts expensive computation offline:

\begin{itemize}
\item Other methods \cite{DBLP:conf/tacas/LiSSS14,DBLP:journals/fmsd/KonighoferABHKT17,DBLP:conf/ecrts/KimVBKLS99,DBLP:conf/hybrid/SankaranarayananF12,DBLP:conf/rv/DokhanchiHF14,DBLP:journals/fmsd/DeshmukhDGJJS17,DBLP:journals/acta/BloemCGHHJKK14} start from discrete specifications and leave the \emph{continuous dynamics implicit and unchecked}.
We, in contrast, start from \emph{hybrid systems models with continuous dynamics} and therefore characterize monitors provably correct with respect to the dynamics model. 
Additionally, unlike \cite{DBLP:conf/rv/DesaiDS17,DBLP:journals/corr/abs-1808-07921,DBLP:conf/hybrid/RoehmOWA16,DBLP:journals/fmsd/MitschP16}, we model hybrid systems with sensors and actuators and therefore:
\begin{inparaenum}[(i)]
\item detect when uncertainty accumulates to unsafe deviations, and
\item can distinguish between violations caused by uncertainty in our own system (sensors, actuators) vs. unsafe environments, which makes our approach better suited to dynamic environments.
\end{inparaenum}
\item
Methods that include discrete models (\eg, \cite{DBLP:conf/ecrts/KimVBKLS99,DBLP:conf/memocode/LiDS11,DBLP:conf/fmcad/AlurMT13,DBLP:conf/tacas/LiSSS14,DBLP:conf/rv/DesaiDS17,DBLP:journals/fmsd/KonighoferABHKT17,DBLP:conf/rv/DokhanchiHF14,DBLP:journals/fmsd/DeshmukhDGJJS17,DBLP:journals/acta/BloemCGHHJKK14}) would require additional assumptions on the continuous dynamics between sampling points in order to be sound \cite{DBLP:conf/hybrid/PlatzerC07}, which we handle explicitly like ModelPlex \cite{DBLP:journals/fmsd/MitschP16}, but characterize partial controllability and partial observability and distinguish between deviation caused by mere uncertainty vs.\ actually unsafe environment behavior.
\item
Some methods \cite{DBLP:conf/hybrid/RoehmOWA16} rely on extensive runtime computations (\eg, reachable sets, whose runtime is hard to predict).
We, in contrast, perform expensive computations offline. At runtime, we only evaluate the resulting formula in real arithmetic for concrete sensor values and control decisions, which enables fast enough responses.
\end{itemize}

Crucially, we \emph{prove correctness properties that correctly link satisfied monitors to offline safety proofs}, result in monitors that \emph{warn ahead of time}, and \emph{account for partial controllability and partial observability} so that the monitored system inherits the provable safety guarantees about the model despite the inevitable presence of uncertainty.

\section{Conclusion}
\label{sec:conclusion}

Provable guarantees about the safety of cyber-physical systems at runtime are crucial as systems become increasingly autonomous. 
Formal verification techniques provide an important basis by proving safety of CPS models, which then requires \emph{transferring the guarantees of offline proofs} to system execution.
We answer the key question of \emph{how offline proofs transfer by runtime monitoring}, and, crucially, \emph{what property needs to be runtime-monitored from sensor measurements} to provably imply safety of the monitored system. 
Our techniques significantly extend previous methods to models of practical interest by characterizing monitors in differential dynamic logic and implementing proof tactics that correctly synthesize monitor conditions that are robust to bounded sensor uncertainty and bounded actuator disturbance, which are the two most fundamental sources of \emph{partial observability} in CPS.

\bibliographystyle{splncs03}      
\bibliography{modelplexsandbox}   

\appendix

\section{Proofs}\label{app:proofs}

This section lists all proofs for the theorems, propositions, and corollaries of this paper.

These proofs use the following characteristics of hybrid programs and \dL formulas:
\begin{inparaenum}[(i)]
\item The truth of a formula $\phi$ only depends on its free variables $\freevars{\phi}$;
\item Hybrid programs only change their bound variables $\boundvars{\alpha}$ but not the complement $\scomplement{\boundvars{\alpha}}$;
\item Similar states (that agree on the free variables) have similar transitions according to the transition relation $\iaccess[\cdot]{}$ of hybrid programs.
\end{inparaenum}
We use $\vars{\alpha} = \freevars{\alpha} \cup \boundvars{\alpha}$ to denote the set of all variables of $\alpha$.
Specifically, our proofs use the following versions of
\cite[Lemmas 9, 11, 12 with 17]{DBLP:journals/jar/Platzer17}.
Hybrid programs only change their bound variables $\boundvars{\alpha}$ but not any other $\scomplement{\boundvars{\alpha}}$:
\begin{lemma}[Bound effect lemma \cite{DBLP:journals/jar/Platzer17}] \label{lem:bound}
If $\iaccessible[\alpha]{\I}{\It}$, then $\iget[state]{\I}=\iget[state]{\It}$ on $\scomplement{\boundvars{\alpha}}$.
\end{lemma}

\noindent The truth of a formula $\phi$ only depends on its free variables $\freevars{\phi}$:
\begin{lemma}[Coincidence lemma \cite{DBLP:journals/jar/Platzer17}]
\label{lem:coincidence}
If $\iget[state]{\I}=\iget[state]{\tI}$ on $\freevars{\phi}$ then $\imodels{\I}{\phi}$ iff $\imodels{\tI}{\phi}$.
\end{lemma}

\noindent Similar states (that agree on the free variables) have similar transitions according to the transition relation of hybrid programs:
\begin{lemma}[Coincidence lemma \cite{DBLP:journals/jar/Platzer17}]
\label{lem:coincidence-HP}
If $\iget[state]{\I}=\iget[state]{\tI}$ on $\hpvars \supseteq \freevars{\alpha}$ and $\iaccessible[\alpha]{\I}{\It}$, then there is a $\iget[state]{\tIt}$ such that $\iaccessible[\alpha]{\tI}{\tIt}$ and $\iget[state]{\It}=\iget[state]{\tIt}$ on $\hpvars$.
\end{lemma}

\noindent Logical state relation $\ddiamond{\alpha}{\Upsilon^+}$ from \cite[Def. 3]{DBLP:journals/fmsd/MitschP16} captures runs of hybrid programs $\alpha$ \cite{DBLP:journals/fmsd/MitschP16}.

\begin{lemma}[Logical state relation \cite{DBLP:journals/fmsd/MitschP16}]\label{lem:logicalrelation}
  Let \(\allvars\) be the set of all variables.
  Two states $\stold,\stnew$ that agree on $\allvars \setminus \boundvars{\alpha}$, \ie, $\stold(z)=\stnew(z)$ for all $z\in\allvars\setminus \boundvars{\alpha}$,
  satisfy
  \(\iaccessible[\alpha]{\I}{\It}\)
  iff
  \(\irelmodels{\I}{\It}{\ddiamond{\alpha}{\Upsilon^+}}\).
\end{lemma}

\printproofs

\section{Derive Monitor Condition in Arithmetical Form}
\label{app:arithmeticalform}

This section gives a simple example illustrating how an (executable) arithmetical formula can be obtained from a monitor characterization \cite{DBLP:journals/fmsd/MitschP16}.

The monitor condition $\didia{\pchoice{\humod{a}{a+1}}{\humod{b}{*};\ptest{b\leq 3}}}(a^+=a\land b^+=b)$ for program $\pchoice{\humod{a}{a+1}}{\humod{b}{*};\ptest{b\leq 3}}$ is turned into arithmetical form with a \dL proof \cite{DBLP:journals/jar/Platzer17}.

\vspace{\baselineskip}
\begin{sequentdeduction}[array]
\linfer[existsr]
{\lsequent{}{(a^+=a+1 \land b^+=b) \lor (a^+=a \land b^+\leq 3)}}
{
\linfer[randomd+testd]
{\lsequent{}{(a^+=a+1 \land b^+=b) \lor \exists b{\leq}3\,(a^+=a \land b^+=b)}}
{
\linfer[assignd]
{\lsequent{}{(a^+=a+1 \land b^+=b) \lor \didia{\humod{b}{*};\ptest{b \leq 3}}(a^+=a \land b^+=b)}}
{
\linfer[choiced]
{\lsequent{}{\didia{\humod{a}{a+1}}(a^+=a \land b^+=b) \lor \didia{\humod{b}{*};\ptest{b \leq 3}}(a^+=a \land b^+=b)}}
{\lsequent{}{\didia{\pchoice{\humod{a}{a+1}}{\humod{b}{*};\ptest{b\leq 3}}}(a^+=a\land b^+=b)}}
}
}
}
\end{sequentdeduction}

The resulting monitor formula $a^+=a+1\land b^+=b \lor a^+=a \land b^+ \leq 3$ of the sequent proof means that either the output $a^+$ is the input $a$ incremented by $1$ while $b$ stayed unchanged, or that the output $b^+\leq 3$ while $a$ stayed unchanged.
The monitor formula is satisfied over states $\stold$ and $\stnew$ with $\stold(a)=2, \stold(b)=3$ and $\stnew(a)=3, \stnew(b)=3$, \ie $\irelmodels{\I}{\It}{a^+=a+1\land b^+=b \lor a^+=a \land b^+ \leq 3}$; it is violated on $\stnew(a)=2, \stnew(b)=4$.

\section{Measurement Modeling Patterns}

In safety proofs, a useful modeling pattern for representing measurements takes measurements before the controller $\measure{y}{\Delta};~\senseCtrl;~\plant{\ctrlOut}$, \eg, as used for modeling speed and position sensors of ground robots in \cite{DBLP:journals/ijrr/MitschGVP17}.
That way, the loop invariants $\inv$ in a safety proof are less cluttered with information about the measurements, which are of temporary nature for making a control decision.
For the purpose of deriving monitoring conditions as introduced in \rref{sec:sensoruncertainty}, however, it is beneficial to have access to a pair of measurements in the loop body, so the program shape becomes $\ctrl{\hat{y}};~\plant{\ctrlOut};~\measure{y}{\Delta}$.
In order to avoid duplicate safety analyses, \rref{lem:measurementrefactoring} provides a way of transferring invariant properties between these two shapes.

\begin{restatable}[Measurement rollover]{lemma}{lem:measurementrefactoring}
\label{lem:measurementrefactoring}
Assume formula \[\init\limply\dibox{\prepeat{(\measure{y}{\Delta};~\senseCtrl;~\plant{\ctrlOut})}}\safe\] is proven with invariant $\inv$, \ie, $\init \limply \inv$, $\inv \limply \dibox{\measure{y}{\Delta};~\senseCtrl;~\plant{\ctrlOut}}\inv$, and $\inv \limply \safe$ are valid.
The measurement $\hat{y}$ is bound only in $\measure{y}{\Delta}$ but nowhere else and assume $\hat{y} \not\in \freevars{\inv}$.
Then, measurement after $\plant{\ctrlOut}$ is equivalent to measurement before $\senseCtrl$:
\begin{align*}
&\underbrace{\inv \limply \dbox{(\measure{y}{\Delta};~\senseCtrl;~\plant{\ctrlOut}}\inv}_{F} \\
\mequiv~ &
\underbrace{\inv \land \nsimilar{\hat{y}}{y}{\Delta} \limply \dbox{(\senseCtrl;~\plant{\ctrlOut};~\measure{y}{\Delta})}\inv}_{G}
\end{align*}
\end{restatable}

\irlabel{approximate|$[:*] \limply [:=]$}
\irlabel{ghost|$[:=]~\text{ghost}$}
\irlabel{BR|$\text{BR}$}
\irlabel{imply1|$P\limply (Q\limply P)$}
\irlabel{allvacuous|$\forall x P \lbisubjunct P ~(x{\not\in}P)$}
\begin{proof}
The proof uses the axioms $\irref{randomb}$ ($\dbox{\humod{x}{*}}P \lbisubjunct \forall x\,P$) and $\irref{testb}$ ($\dbox{\ptest{Q}}P \lbisubjunct (P\limply Q)$) \cite{DBLP:journals/jar/Platzer17} together with propositional reasoning and bound renaming to split off and introduce measurements:
\begin{description}
\item[$\limply$] First, the test $\irref{testb}$ and random assignment $\irref{randomb}$ split off the measurement after plant $\plant{\ctrlOut}$, then $\irref{testb}$ and $\irref{randomb}$ are applied in the inverse direction to introduce the measurement before $\senseCtrl$:
\begin{sequentdeduction}[array]
\linfer[id]
{\lclose}
{
\linfer[testb+randomb+composeb+weakenl]
{\lsequent{F}{\inv \limply \dbox{\measure{y}{\Delta};~\senseCtrl;~\plant{\ctrlOut}}\inv}}
{
\linfer[allvacuous]
{\lsequent{F}{\inv \land \nsimilar{\hat{y}}{y}{\Delta} \limply \dbox{\senseCtrl;~\plant{\ctrlOut}}\inv}}
{
\linfer[composeb+randomb+testb]
{\lsequent{F}{\inv \land \nsimilar{\hat{y}}{y}{\Delta} \limply \dbox{\senseCtrl;~\plant{\ctrlOut}}(\forall \nsimilar{\hat{y}}{y}{\Delta}\,\inv)}}
{
\lsequent{F}{
\inv \land \nsimilar{\hat{y}}{y}{\Delta} \limply \dbox{\senseCtrl;~\plant{\ctrlOut};~\measure{y}{\Delta}}\inv}
}
}
}
}
\end{sequentdeduction}
\item[$\leftarrow$] The converse direction introduces an exact measurement $\humod{\hat{y}^+}{y}$ with a fresh ghost variable $\hat{y}^+$ first and then introduces measurement error using \irref{approximate}, which is defined as $\dbox{\underbrace{\prandom{\hat{y}};\ptest{y-\Delta \leq \hat{y} \leq y+\Delta}}_{\measure{y}{\Delta}}}P \limply \dbox{\humod{\hat{y}}{y}}P$:
\begin{sequentdeduction}[array]
\linfer[id]
{\lclose}
{
\linfer[BR]
{\lsequent{G}{\inv \land \nsimilar{\hat{y}}{y}{\Delta} \limply \dbox{\senseCtrl;~\plant{\ctrlOut};~\measure{y}{\Delta}}\inv}}
{
\linfer[approximate]
{\lsequent{G}{\inv \land \nsimilar{\hat{y}}{y}{\Delta} \limply \dbox{\senseCtrl;~\plant{\ctrlOut};~\measure{y}{\Delta}}\inv}}
{
\linfer[ghost+composeb]
{\lsequent{G}{\inv \land \nsimilar{\hat{y}}{y}{\Delta} \limply \dbox{\senseCtrl;~\plant{\ctrlOut};~\humod{\hat{y}^+}{y}}\inv}}
{
\linfer[randomb+allr+testb]
{\lsequent{G}{\inv \land \nsimilar{\hat{y}}{y}{\Delta} \limply \dbox{\senseCtrl;~\plant{\ctrlOut}}\inv}}
{
\lsequent{G}{\inv \limply \dbox{\measure{y}{\Delta};~\senseCtrl;~\plant{\ctrlOut}}\inv}
}
}
}
}
}
\end{sequentdeduction}
\end{description}
\end{proof}

\rref{lem:measurementrefactoring} allows us to switch between the measurement modeling patterns for safety proofs and model monitors.

\section{Proof-Guided Model Monitors for Nonlinear Dynamics}

Existing techniques for model monitor synthesis \cite{DBLP:journals/fmsd/MitschP16} require symbolic closed-form polynomial solutions to characterize differential equations.
Hybrid systems with nonlinear dynamics, however, do not necessarily have symbolic closed-form solutions, or their solutions are not expressible in first-order real arithmetic.
In such cases, hybrid systems safety proofs employ a more general technique based on invariant properties of differential equations \cite{DBLP:conf/lics/PlatzerT18} to abstract away the concrete trajectories of nonlinear differential equations by describing \emph{invariant regions} that confine the trajectories.
The challenge with such overapproximations is that they are \emph{not enough} to conclude the existence of a model run as required for correct monitor synthesis \cite{DBLP:journals/fmsd/MitschP16}. 

The crucial insight that we gain from a successful safety proof for monitoring and model validation is that the potentially complicated exact behavior is not relevant in detail for safety, but that it was enough to stay inside a safety-relevant region.
We exploit this observation to synthesize monitoring conditions that only check these safety-critical invariant regions and thereby allow for a wider safe variety of actual system behavior instead of insisting on the specific modeled behavior.

In this section, we discuss techniques to translate models to make sure that their differential invariants are expressed explicitly and their specific dynamics are abstracted to any behavior inside the invariant regions. 
That way, the invariant conditions are picked up during monitor synthesis in a provably correct way and become represented in the monitoring conditions that are derived from these models.
As a main insight, we exploit the fact that the desired monitoring conditions must concisely capture paths to safety violation. 
Specifically, a model monitor for a model $\alpha$ of the form $\ctrlPrg;\plantPrg$ uses a model $\tilde{\alpha}$ that overapproximates the plant $\plantQ{\ctrlOut}$ with nondeterministic assignments $\prandom{x}$. 
These assignments are guarded by the evolution domain $\ivr$ and differential invariants $R(x,x_0)$ before and after the nondeterministic assignments to conservatively preserve the semantics of evolution domain constraints in differential equations, so $\tilde{\alpha}$ has the shape $\humod{x_0}{x};\ptest{\ivr};\prandom{x};\ptest{(\ivr \land R(x,x_0)}$.

\rref{thm:nondetplantimpliessafety} guarantees that a satisfied monitor for program $\tilde{\alpha}$ preserves an inductive safety property for one control step, so can be extended to loops in a straightforward way as in \cite{DBLP:journals/fmsd/MitschP16}. 
For simplicity, we assume that the world outside $\tilde{\alpha}$ is unmodified on $\scomplement{\boundvars{\tilde{\alpha}}}$, which can be lifted easily with techniques in \cite{DBLP:journals/fmsd/MitschP16}.

\begin{theorem}[Nonlinear model monitor correctness]
\label{thm:nondetplantimpliessafety}
Let $\alpha$ be a hybrid program of the form $\ctrl{x};\plant{\ctrlOut}$ with evolution domain constraint $\ivr$.
Let $\prepeat{\alpha}$ be provably safe with invariant $\inv$, so $A \limply \dbox{\prepeat{\alpha}}S$, $A\limply \inv$, $\inv \limply\dbox{\alpha}\inv$, and $\inv\limply S$ are valid.
Let $x=\boundvars{\plant{\ctrlOut}}$ and let $x_0$ be fresh variables not in $\vars{\alpha}$.
Let the differential invariants $R(x,x_0)$ be provable, so $\inv \limply \dbox{\ctrl{x};\humod{x_0}{x};\plant{\ctrlOut}}R(x,x_0)$ is valid.
Assume the system transitions from state $\stold$ to state $\stnew$, which agree on $\scomplement{\boundvars{\alpha}}$, and assume $\stold \models \inv$.
If the nonlinear model monitor 
\[\chimpp \equiv \didia{\prepeat{\bigl(\ctrl{x};\humod{x_0}{x};\ptest{\ivr};\prandom{x};\ptest{(\ivr \land R(x,x_0)\bigr)}}}\Upsilon^+\]
is satisfied, \ie, $\irelmodels{\I}{\It}{\chimpp}$,
then the invariant $\inv$ is preserved, \ie, $\imodels{\It}{\inv}$.
\end{theorem}

\begin{proof}
Follows from \cite[Theorem 1]{DBLP:journals/fmsd/MitschP16} by reducing the complicated dynamics $\inv \limply \dbox{\left(\ctrl{x};\{\D{x}=\theta ~\&~ \ivr\}\right)}\inv$ to the safety proof of the nondeterministic approximation, \ie, formula \eqref{eq:nondetoverapproximation} is valid.
\begin{equation}\label{eq:nondetoverapproximation}
\begin{aligned}
&\Bigl(\underbrace{\inv \limply \dbox{\ctrl{x};\humod{x_0}{x};\ptest{\ivr};\prandom{x};\ptest{\ivr \land R(x,x_0)}}\inv}_{C}\Bigr)\\
&\limply \Bigl(\inv \limply \dbox{\ctrl{x};\{\D{x}=\theta ~\&~ \ivr\}}\inv\Bigr)
\end{aligned}
\end{equation}

\irlabel{GVR|$\text{GVR}$}
\irlabel{loopify|$[*] \text{intro}$}
\irlabel{DQinitial|$[?\ivr]$}
\irlabel{randomb|$[{:}{*}]$}

\begin{sequentdeduction}[array]
\linfer[id]
{\lclose}
{
\linfer[randomb+composeb]
{
\lsequent
{C}
{\inv \limply \dbox{\ctrl{x};~\humod{x_0}{x};\ptest{\ivr};\prandom{x};\ptest{\ivr\land R(x,x_0)}}\inv }
}
{
\linfer[GVR]
{
\lsequent
{C}
{\inv \limply \dbox{\ctrl{x};~\humod{x_0}{x};\ptest{\ivr}}\forall x\dbox{\ptest{\ivr\land R(x,x_0)}}\inv }
}
{
\linfer[composeb]
{
\lsequent
{C}
{\inv \limply \dbox{\ctrl{x};~\humod{x_0}{x};~\ptest{\ivr}}\dbox{\plantQ{\ctrlOut}\land R(x,x_0)}\dbox{\ptest{\ivr\land R(x,x_0)}}\inv }
}
{
\linfer[DW+testb]
{
\lsequent
{C}
{\inv \limply \dbox{\ctrl{x};~\humod{x_0}{x};~\ptest{\ivr};~\plantQ{\ctrlOut}\land R(x,x_0)}\dbox{\ptest{\ivr\land R(x,x_0)}}\inv }
}
{
\linfer[dC]
{
\lsequent
{C}
{\inv \limply \dbox{\ctrl{x};~\humod{x_0}{x};~\ptest{\ivr};~\plantQ{\ctrlOut} \land R(x,x_0)}\inv \hfill\triangleright}
}
{
\linfer[DQinitial]
{
\lsequent
{C}
{\inv \limply \dbox{\ctrl{x};~\humod{x_0}{x};~\ptest{\ivr};~\plantQ{\ctrlOut}}\inv }
}
{
\linfer[ghost]
{
\lsequent
{C}
{\inv \limply \dbox{\ctrl{x};~\humod{x_0}{x};~\plantQ{\ctrlOut}}\inv }
}
{
\lsequent
{C}
{\inv \limply \dbox{\ctrl{x};~\plantQ{\ctrlOut}}\inv}
}
}
}
}
}
}
}
}
\end{sequentdeduction}

We use $\irref{ghost}$ to introduce assignments to fresh variables $x_0$ that store the state of $x$ before the plant, so that it is available when describing the differential equation with the differential invariant $R(x,x_0)$ in step $\irref{dC}$.
The side condition closes by assumption $\inv \limply \dbox{\ctrl{x};\humod{x_0}{x};\plant{\ctrlOut}}R(x,x_0)$.
Differential invariants hold throughout the entire evolution of a differential equation, from beginning to end: step $\irref{DW},\irref{testb}$ makes it available after the differential equation, step $\irref{DQinitial}$ at the beginning using the axiom $\dbox{\ptest{q(x)}}\dbox{\{\D{x}=f(x)\&q(x)\}}p(x) \lbisubjunct \dbox{\{\D{x}=f(x)\&q(x)\}}p(x)$ derived from $\irref{DI}$ \cite{DBLP:journals/jar/Platzer17}.
Now that the differential equation is safeguarded by its differential invariants, in step $\irref{GVR}$ and $\irref{randomb}$ we abstract from the concrete dynamics by allowing any evolution that satisfies the differential invariants.
\end{proof}

\begin{mdframed}[style=example]
\begin{example}[Model monitor]
The flight protocol with a nondeterministic plant turns into the following monitor condition, where $\Upsilon^+$ is $\bigwedge_{z\in\{\xrel,\yrel,\theta,\wownship\xrel_0,\yrel_0,\theta_0\}}z^+{=}z$:
\begin{multline*}
\bigl\langle \ctrlPrg;~
\humod{(\xrel_0}{\xrel};\humod{\yrel_0}{\yrel};\humod{\theta_0}{\theta};~
\ptest{\top};
\prandom{\xrel};\prandom{\yrel};\prandom{\theta};
\\
\ptest{(\wownship{=}0 \limply \straightinvlhs=\straightinvlhs(0)) \land (\wownship{=}1 \limply \flightinvlhs=\flightinvlhs(0))}\Upsilon^+
\end{multline*}
\noindent with 
\begin{align*}
\straightinv &\equiv \underbrace{\straightinvlhsdef}_{\straightinvlhs} > \vownship+\vintruder \qquad \text{and}\\ 
\flightinv & \equiv \underbrace{\flightinvlhsdef{\wangular}}_{\flightinvlhs} > \vownship\vintruder + \vintruder \wangular
\end{align*}

\noindent The synthesis steps in \rref{sec:evaluation} (\rref{fig:modelmonitorpreprocessing}) produce this arithmetical monitor condition:
\begin{multline*}
\straightinv \land \theta^+{=}\theta \land \sintheta^+\xrel^+ - (\costheta^+-1)\yrel^+ = \sintheta\xrel - (\costheta-1)\yrel \land \wownship^+{=}0 \land \wintruder^+{=}0\\
\lor
\flightinv{\wownship^+} \land \sintheta^+\xrel^+ - \costheta^+\yrel^+ + \costheta^+ = \sintheta\xrel - \costheta\yrel + \costheta \land \wownship^+{=}1 \land \wintruder^+{=}0
\end{multline*}
\end{example}
\end{mdframed}

The model monitors derived by abstracting from the specific continuous dynamics to the invariant regions identified in the safety proof allow monitoring of models without symbolic closed-form solution.
The invariant regions allow for some deviation from the specific modeled dynamics, but still assume that control choices are perfectly turned into physical effects and that the real world is perfectly observable.
The techniques in the main part of this paper present two fundamental extensions with different flavors of uncertainty: actuator disturbance due to partial controllability and sensor uncertainty due to partial observability.

\section{Implementation}
\label{app:implementation}

\begin{figure}[htb]
\normalsize
\begin{equation*}
\rotatebox[origin=c]{90}{\begin{minipage}{3.5cm}\centering \textbf{Offline}\\[-.5ex]transformation by proof\end{minipage}}
\left\downarrow
\begin{aligned}
\text{Safety} && && A &\limply \dbox{\prepeat{\alpha}}S \\
\text{Safety} && \raisebox{2ex}[0ex][0ex]{by \rref{thm:nondetplantimpliessafety} $\Uparrow$} && A &\limply \dbox{\prepeat{\tilde{\alpha}}}S\\
\text{Semantical} && \raisebox{2ex}[0ex][0ex]{by \cite[Theorem 1]{DBLP:journals/fmsd/MitschP16} $\Uparrow$} && (\stold,\stnew) &\in \iaccess[\prepeat{\tilde{\alpha}}]{} \\
\text{Logical} && \raisebox{2ex}[0ex][0ex]{by \cite[Lemma 4]{DBLP:journals/fmsd/MitschP16} $\Updownarrow$} && (\stold,\stnew) &\models \didia{\prepeat{\tilde{\alpha}}}\Upsilon^+ \\
\text{Logical} && \raisebox{2ex}[0ex][0ex]{by \cite[Lemma 5]{DBLP:journals/fmsd/MitschP16} $\Uparrow$} && (\stold,\stnew) &\models \didia{\tilde{\alpha}}\Upsilon^+\\
\text{Logical} && \raisebox{2ex}[0ex][0ex]{by \rref{thm:mm-pairwisecorrectness} $\Uparrow$} && (\stold,\stnew) &\models \exists y{\in}\neighborhood{\hat{y}}{\Delta}\didia{\tilde{\alpha}}(\exists y^+,\Upsilon^+)\\
\text{Arithmetical}&& \raisebox{2ex}[0ex][0ex]{by \dL proof $\Uparrow$} && (\stold,\stnew) &\models \fchisp \textbf{ by online monitoring}
\end{aligned}
\right.
\end{equation*}
\caption{Proof-guided model transformation for synthesizing a monitor for nonlinear differential equations with sensor uncertainty by monitoring for pairwise existence of unobservable true values $y$, $y^+$.
Synthesis correct by a chain of semantical representation, logic characterization, and arithmetical form of a monitor.}
\label{fig:modelmonitorpreprocessing}
\end{figure}

The process for synthesizing monitors from hybrid systems models with actuator disturbance and sensor uncertainty is illustrated in \rref{fig:modelmonitorpreprocessing}, with original model $\alpha$ and overapproximated model $\tilde{\alpha}$.
The process is implemented as a tactic in \KeYmaeraX \cite{DBLP:conf/cade/FultonMQVP15}.
The bottom step in \rref{fig:modelmonitorpreprocessing} uses \dL automation and quantifier elimination (QE) to synthesize an easily executable quantifier-free arithmetical model monitoring condition $F(x,x^+)$ as illustrated in the following sequent proof.
The steps are a straightforward application of \dL axioms from the innermost formula going outwards.
The existential quantifier can be instantiated by applying \cite[Opt. 1]{DBLP:journals/fmsd/MitschP16}, since $\Upsilon^+\equiv x^+{=}x$.
Any remaining quantified sub-formulas are turned into their equivalent quantifier-free form using preprocessing by lemmas \eqref{eq:existsOr} and \eqref{eq:existsVacuous} with external solvers for QE, which are connected to \KeYmaeraX (\eg, Mathematica). 

\irlabel{randomd|$\didia{{:}{*}}$}

\begin{sequentdeduction}[array]
\linfer[RCFp]
{\lsequent{}{F(x\setminus y,x^+\setminus y^+)}}
{
\linfer[existsr]
{
\lsequent{}{
\exists y{\in}\neighborhood{\hat{y}}{\Delta} \Bigl(\ivr(\ctrlOut) \land \ivr(x^+)\land R(x^+,\ctrlOut)\land(\exists y^+\, \Upsilon^+)\Bigr)}
}
{
\linfer[assignd]
{
\lsequent{}{
\exists y{\in}\neighborhood{\hat{y}}{\Delta} \Bigl(\ivr(\ctrlOut)\land \exists x\, \bigl(\ivr\land R(x,\ctrlOut)\land(\exists y^+\, \Upsilon^+)\bigr)\Bigr)}
}
{
\linfer
{
\lsequent{}{
\exists y{\in}\neighborhood{\hat{y}}{\Delta}\ddiamond{\humod{x_0}{\ctrlOut}}\Bigl(\ivr(\ctrlOut) \land \exists x\, \bigl(\ivr\land R(x,x_0)\land(\exists y^+\, \Upsilon^+)\bigr)\Bigr)}
}
{
\linfer[testd]
{
\lsequent{}{
\exists y{\in}\neighborhood{\hat{y}}{\Delta}\ddiamond{\senseCtrl}\ddiamond{\humod{x_0}{x}}\Bigl(\ivr\land \exists x\, \bigl(\ivr\land R(x,x_0)\land(\exists y^+\, \Upsilon^+)\bigr)\Bigr)}
}
{
\linfer[randomd]
{
\lsequent{}{
\exists y{\in}\neighborhood{\hat{y}}{\Delta}\ddiamond{\senseCtrl}\ddiamond{\humod{x_0}{x}}\ddiamond{\ptest{\ivr}}\exists x\, \bigl(\ivr\land R(x,x_0)\land(\exists y^+\, \Upsilon^+)\bigr)}
}
{
\linfer[testd]
{
\lsequent{}{
\exists y{\in}\neighborhood{\hat{y}}{\Delta}\ddiamond{\senseCtrl}\ddiamond{\humod{x_0}{x}}\ddiamond{\ptest{\ivr}}\ddiamond{\prandom{x}}\bigl(\ivr\land R(x,x_0)\land(\exists y^+\,\Upsilon^+)\bigr)}
}
{\lsequent{}{\exists y{\in}\neighborhood{\hat{y}}{\Delta}\ddiamond{\senseCtrl}\ddiamond{\humod{x_0}{x}}\ddiamond{\ptest{\ivr}}\ddiamond{\prandom{x}}\ddiamond{\ptest{\ivr\land R(x,x_0)}}(\exists y^+\,\Upsilon^+)}
}
}
}
}
}
}
}
\end{sequentdeduction}

\KeYmaeraX applies time-bounded heuristics to reduce formula size before calling external solvers, since the performance of external solvers in this final step depends on the number of variables and quantifier alternations in the monitor condition.
As a result, the proof step numbers reported here may slightly differ when rerunning the synthesis on other platforms.

\end{document}